\documentclass[12pt]{amsart}

\usepackage{amssymb,latexsym,amsfonts, amsthm, amsmath, amsrefs}
\usepackage{enumerate}
\usepackage{array}
\usepackage{booktabs}
\usepackage{fancyhdr}
\usepackage{verbatim}
\usepackage{multicol}
\usepackage{longtable}
 \usepackage{tikz}
\usepackage[all,cmtip]{xy}
\usetikzlibrary{arrows,chains,matrix,positioning,scopes}
\makeatletter
\tikzset{join/.code=\tikzset{after node path={%
\ifx\tikzchainprevious\pgfutil@empty\else(\tikzchainprevious)%
edge[every join]#1(\tikzchaincurrent)\fi}}}
\tikzset{>=stealth',every on chain/.append style={join},
         every join/.style={->}}

\input xy
\xyoption{all}

\newcommand{\bcat}{$\beta$-catenin }

\newtheorem{thm}{Theorem}[section]   

\newtheorem{cor}[thm]{Corollary}     
\newtheorem{lemma}[thm]{Lemma}         
\newtheorem{prop}[thm]{Proposition}  

\theoremstyle{definition} 
\newtheorem{defn}[thm]{Definition}   
\newtheorem{remark}[thm]{Remark}   
\newtheorem{ex}[thm]{Example}        
\newtheorem{rmk}[thm]{Remark}

\newcommand{\RR}{\ensuremath{\mathbb{R}}}

\newcommand{\CC}{\ensuremath{\mathbb{C}}}

\usepackage{amssymb}
\usepackage{fullpage}
\usepackage{parskip}
\usepackage{color}
\usepackage{soul}

\title{Algebraic systems biology: \\
a case study for the Wnt pathway}
\author{Elizabeth Gross, Heather A. Harrington, Zvi Rosen, and Bernd Sturmfels}
\address{Elizabeth Gross: San Jos\'e State University,
{\tt elizabeth.gross@sjsu.edu} \hfill \break
Heather A.~Harrington: University of Oxford, 
{\tt harrington@maths.ox.ac.uk} \hfill \break
Zvi Rosen: University of California at Berkeley,
{\tt zhrosen@berkeley.edu} \hfill \break
Bernd Sturmfels: University of California at Berkeley,
{\tt bernd@berkeley.edu}
}

\begin{document}
\setlength{\parskip}{5pt}

\begin{abstract} 
Steady state analysis of dynamical systems for biological networks
give rise to algebraic varieties in high-dimensional spaces whose study 
is of interest in their own right. We demonstrate this for the shuttle model of the 
Wnt signaling pathway. Here the variety is described by a polynomial system in  
$19$ unknowns and $36$ parameters. Current methods from computational 
algebraic geometry and combinatorics are applied to analyze this model.
\end{abstract}

 \maketitle

\section{Introduction}

The theory of biochemical reaction networks is fundamental for
 systems biology \cites{Klipp, Voit}.  It is based on a wide range of 
mathematical fields, including dynamical systems,  numerical analysis,
optimization, combinatorics,  probability, and,
last but not least, algebraic geometry. 
There are numerous articles that use algebraic geometry in the study of
biochemical reaction networks, especially those arising from mass action kinetics.
A tiny selection is \protect{\cite{CF,FW,KPDDG,PDSC,ShSt}}.

We here perform a detailed analysis of one specific system,
namely the shuttle model for the Wnt signaling pathway, 
introduced recently by MacLean, Rosen, Byrne, and Harrington~\cite{MRBH}.
Our aim is twofold: to demonstrate how biology can lead to interesting
questions in algebraic geometry and to apply state-of-the-art techniques from computational algebra
to biology.

The dynamical system we study consists of the following $19$ ordinary differential equations.
Their derivation and  the relevant background from biology will be presented in Section~\ref{sec:bio}.
\begin{equation}
 \label{eq:diffeqn}
\begin{matrix}
	  {\dot{x}_{1}}		&=& -k_{1} x_{1} + k_{2} x_{2}  \\
	 {\dot{x}_{2}} 	    	&=&  k_{1} x_{1} - (k_{2}+ k_{26}) x_{2} + k_{27} x_{3} - k_{3} x_{2} x_{4} +  (k_{4}+k_{5}) x_{14}  \\
	 {\dot{x}_{3}} 		&= &  k_{26} x_{2} - k_{27} x_{3} - k_{14} x_{3} x_{6} + (k_{15}+   k_{16}) x_{15} \\
	 {\dot{x}_{4}} 		&=& -k_{3} x_{2} x_{4} - k_{9} x_{4} x_{10} + k_{4} x_{14} + k_{8} x_{16} +  ( k_{10} + k_{11}) x_{18} \\
	 {\dot{x}_{5}} 		&= & -k_{28} x_{5} + k_{29} x_{7} - k_{6} x_{5} x_{8} + k_{5} x_{14} + k_{7} x_{16} \\
	 {\dot{x}_{6}} 		&= & -k_{14} x_{3} x_{6} - k_{20} x_{6} x_{11} + k_{15} x_{15} +  k_{19} x_{17} + (k_{21} + k_{22}) x_{19} \\
	 {\dot{x}_{7}}		&=& k_{28} x_{5} - k_{29} x_{7} - k_{17} x_{7} x_{9} + k_{16} x_{15} +  k_{18} x_{17} \\
	 {\dot{x}_{8}} 		\,\,=\,  - {\dot{x}_{16}} & = &
	 -k_{6} x_{5} x_{8} + (k_{7} + k_{8}) x_{16}  \\
	 {\dot{x}_{9}} \,\, = \, 	- {\dot{x}_{17}}	&=&  -k_{17} x_{7} x_{9} + (k_{18} + k_{19}) x_{17} \\
	 {\dot{x}_{10}}	&=& k_{12} - (k_{13}+k_{30}) x_{10} - k_{9} x_{4} x_{10} + k_{31} x_{11} +  k_{10} x_{18} \\	
	 {\dot{x}_{11}}	&=&  -k_{23}x_{11} + k_{30} x_{10} - k_{31} x_{11} - k_{20} x_{6} x_{11} - k_{24} x_{11} x_{12} +   k_{25} x_{13} + k_{21} x_{19} \\
	 {\dot{x}_{12}}	\,\, = \, -{\dot{x}_{13}} &=& -k_{24} x_{11} x_{12} + k_{25} x_{13}  \\
	 {\dot{x}_{14}} 	&=& k_{3} x_{2} x_{4} - (k_{4} + k_{5}) x_{14}  \\
		 {\dot{x}_{15}}	&=& k_{14} x_{3} x_{6} - (k_{15} + k_{16}) x_{15} \\
	 {\dot{x}_{18}}	&=& k_{9} x_{4} x_{10} - (k_{10} + k_{11}) x_{18} \\
	 {\dot{x}_{19}}	&=& k_{20} x_{6} x_{11} - (k_{21} + k_{22}) x_{19}
	 \end{matrix}
\end{equation}

The quantity $x_i$ is a differentiable function of an unknown $t$, representing time, and ${\dot{x}_i}(t)$ is
the derivative of that function. 
This dynamical system has five linear conservation laws:
\begin{equation}
 \label{eq:conservation}
\begin{matrix}
	 0	 &=& (x_1+x_2+x_3+x_{14}+x_{15})-c_1 \\
	 0 &= & (x_4+x_5+x_6+x_7+x_{14}+x_{15}+x_{16}+x_{17}+x_{18}+x_{19}) -c_2 \\
	 0	&=  & (x_8+x_{16})-c_3 \\
	 0	&= & (x_9+x_{17})-c_4 \\
	 0	  &= & (x_{12}+x_{13}) -c_5
\end{matrix}
\end{equation}
The $31$ quantities $k_i$ are the rate constants of the chemical reactions, and the five $c_i$ 
are the conserved quantities. Both of these are regarded as parameters, so we
have $36$ parameters in total.  Our object of interest is the {\em steady state variety},
which is the common zero set of the right hand sides  of (\ref{eq:diffeqn}) and (\ref{eq:conservation}).
This variety lives in $K^{19}$, where  $K$ is an algebraically closed field 
that contains the rational numbers $\mathbb{Q}$ as well as the $36$ parameters $k_i$ and $c_i$.
If these parameters are fixed to be particular real numbers then we can take
$K = \mathbb{C}$, the field of complex numbers. If it is preferable to regard 
${\bf k} = (k_1,\ldots,k_{31})$ and ${\bf c} = (c_1,\ldots,c_5)$ as vectors of unknowns,
then $K = \overline{ \mathbb{Q}({\bf k},{\bf c})}$ is the algebraic closure
of the rational function field. 
In this latter setting, when all parameters are generic, we shall derive the following result:

\begin{thm}
\label{thm:nine}
The polynomials in 
(\ref{eq:diffeqn})--(\ref{eq:conservation}) have
$9$ distinct zeros in $K^{19}$ when $K = \overline{ \mathbb{Q}({\bf k},{\bf c})}$.
\end{thm}
By analyzing the steady state variety, we can better understand the model, which is nonlinear, 
and thus the biological system. The aim is to predict the system's behavior, offer biological insight, 
and determine what data are required to verify or reject the model. 
Here is a list of questions one might ask about our model from
the perspective of systems biology.

\subsection*{Biological Problems}
These are labeled according to the section that will address them.
\begin{enumerate}[1.]

\item[4.] {\em For what real positive rate parameters and conserved quantities does the system exhibit multistationarity?}  This question is commonly asked when using a dynamical system for modeling a real-world phenomenon. 
When modeling a process that experimentally appears to have more than one stable equilibrium, multistationary models are preferred.

\item[5.] {\em Suppose we can measure only a subset of the species concentrations. Which subsets can lead to model rejection?}
If all species are measurable at steady state, then we can substitute data into the system (\ref{eq:diffeqn}), and check that all expressions ${\dot{x}_i}$
are close to zero.
 If only some $x_i$ are known,  we still want to be able to evaluate models with the available data. 

\item[6.] {\em Give a complete description of the stoichiometric compatibility classes for the chemical reaction network.} A  stoichiometric compatibility class is the set of all points accessible from a given state via 
the reactions in the system. 
This question relates more closely to the dynamics of the system, but also has ramifications for the
set of all steady states.

\item[7.] {\em What information does species concentration data give us for parameter estimation?} \linebreak In particular, are the parameters identifiable? Identifiability  means that 
having many measurements of the concentrations ${\bf x}$ can determine the
reaction rate constants ${\bf k}$. If not identifiable, we will explore algebraic constraints imposed by the species concentration data. This question is relevant for complete and partial steady-state  data (usually noisy).
\end{enumerate}

\smallskip
These questions are open challenges for medium to large models in
systems biology and medicine \cite{Klipp, Voit}.
The book chapter \cite{MHSB} illustrates
standard mathematical and statistical methods for addressing these questions, with Wnt signaling as a
case study. Here, we examine
these questions from the perspective of algebraic geometry.
The aim is to provide insight into global behavior by applying tools from
nonlinear algebra to synthetic and systems
biology. Below are the algebraic problems underlying the four biological problems listed above. 
 
\subsection*{Algebraic Problems}

\begin{enumerate}[1.]

\item[4.]  Describe the set of points $({\bf k},{\bf c}) \in \RR_{>0}^{31} \times \RR_{>0}^5$ such that
the polynomials (\ref{eq:diffeqn})-(\ref{eq:conservation}) have two or more
positive zeros ${\bf x} \in \RR_{>0}^{19}$.
When is there only one? Identify the discriminant.

\item[5.]  Which projections of the variety defined by (\ref{eq:diffeqn}) into coordinate subspaces of $K^{19}$ 
are surjective? Equivalently, describe the algebraic matroid on the ground set  $\{x_1,\ldots,x_{19}\}$.

\item[6.] The conservation relations (\ref{eq:conservation}) specify a linear map $\chi : \RR^{19} \to \RR^5,\,
{\bf x} \mapsto {\bf c}$. Describe all the convex polyhedra
$\,\chi^{-1}({\bf c}) \cap \RR^{19}_{\geq 0}\,$ where ${\bf c}$ runs over the
points in the open orthant $ \RR_{>0}^5$.

\item[7.] \begin{enumerate}[a.] \item {\em Complete data:} Describe the matroid on the ground set 
$\{k_1,k_2,\ldots,k_{31}\}$ that is defined by the linear forms
on the right hand sides of (\ref{eq:diffeqn}), for
fixed steady-state concentrations.
\item {\em Partial steady-state data without noise:} 
Repeat the analysis after eliminating some of the ${\bf x}$-coordinates.
\item {\em Partial steady-state data with noise:} For the remaining ${\bf x}$-coordinates,
suppose that we have data
which are {\em approximately} on the 
projected steady state variety.
Determine a parameter vector $({\bf k},{\bf c})$ that best fits the data. 
\end{enumerate}
\end{enumerate}

In this paper we shall address these questions, and several related ones, after
explaining the various ingredients.
A particular focus is the exchange between the algebraic formulation and its 
biological counterpart.
Our presentation is organized as follows.

In  Section \ref{sec:bio} we review the basics on the Wnt signaling pathway,
we recall the shuttle model of  MacLean {\it et al.}~\cite{MRBH}, and we derive
the dynamical system (\ref{eq:diffeqn})--(\ref{eq:conservation}).
In Section \ref{sec3} we establish Theorem \ref{thm:nine},
and we examine the set of all steady states. This is here regarded
as a complex algebraic variety in an affine space of dimension
$55 = 19+31+5$ with coordinates $({\bf x},{\bf k},{\bf c})$.

In Sections \ref{sec4}, \ref{sec5}, \ref{sec6}, and \ref{sec7} 
we address the four problems stated above. The numbers of the problems
refer to the respective sections. Each section starts out with an explanation
of how the biological problem and the algebraic problem are related.
The rationale behind Section \ref{sec4}  is likely to  be familiar to most of our
readers, given that multistationarity has been discussed widely in the 
literature; see e.g.~\cite{CF,PDSC}.
On the other hand,
in Section~\ref{sec5} we employ the language of matroid
theory. This may be unfamiliar to many readers, especially
 when it comes to the algebraic matroid associated with
an irreducible algebraic variety.  Section~\ref{sec6} characterizes 
the polyhedral geometry encoded in the
conservation relations (\ref{eq:conservation}).
This is a case study
in the spirit of \cite[Figure 1]{ShSt}.
Section~\ref{sec7} addresses
the problems of parameter identifiability and parameter estimation.
Finally, in Section~\ref{sec8} we return to the biology,
and we discuss what our findings might imply for the study of
Wnt signaling and other systems.

\bigskip \bigskip \bigskip

\section{From Biology to Algebra}
\label{sec:bio}

Cellular decisions such as cell division, specialization and cell death are governed by a rich repertoire of complex signals that are produced by other cells and/or stimuli. In order for a cell to come to an appropriate decision, it must \emph{sense} its external environment, communicate this information to the nucleus, and respond by regulating genes and producing relevant proteins. Signaling molecules called ligands, external to the cell, can bind to proteins called receptors, initializing the propagation of information within the cell by molecular interactions and modifications (e.g. phosphorylation). This signal may be relayed from the cytoplasm into the nucleus via molecules and the cell responds by activation or deactivation of gene(s) that control, for example, cell fate. The complex interplay of molecules involved in this information transmission is called a signaling transduction pathway.  Although many signaling pathways have been defined biochemically, much is still not understood about them or how a signal results in a particular cellular response. Mathematical models constructed at different scales of molecular complexity may 
help unravel the central mechanisms that govern cellular decisions, and their analysis may 
inform and guide testable hypotheses and therapies.

In this paper, we focus on the canonical Wnt signaling pathway, which is involved in cellular processes, 
both during development and in adult tissues. This includes stem cells.
Dysfunction of this pathway has been linked to neurodegenerative diseases and cancer.
Consequently, Wnt signaling 
has been widely studied in various organisms, including amphibians and mammals. Researchers are interested in how the extracellular ligand Wnt affects the protein $\beta$-catenin, 
which plays a pivotal role in turning genes on and off in the nucleus.

The molecular interactions within the Wnt signaling pathway are not yet fully understood.
This  has led to the development and analysis of many mathematical models. The Wnt shuttle model \cite{MRBH} includes an abstraction of the signal transduction pathway (via activation/inactivation of molecules) described above. 
The model also takes into account molecules that
exist, interact and move between different compartments  in the cell
 (e.g., cytoplasm and nucleus). Biologists understand the Wnt system as either {\em Wnt off} or {\em Wnt on}.
 However, such a scenario is rarely binary (i.e., different concentration levels of Wnt may exist) and inherently depends on spatial movement of molecules. The Wnt shuttle model 
 includes complex interactions with nonlinearities arising in the equations.
 In particular, it  includes both the Wnt off and Wnt on scenarios,  by adjusting initial conditions or parameter values.  The biology needed to understand the model can be described as follows. See also Table~\ref{tab-notation}.
\par

\textit{Wnt off:} When cells do not sense the extracellular ligand Wnt, \bcat~is degraded (broken down). The degradation of \bcat~is partially dependent on a group of molecules (Axin, APC and GSK-3) that form the {\em destruction complex}.
Crucially, the break down of \bcat~occurs when the destruction complex is in an active state; modification to the destruction complex by proteins, called phosphatases, changes it from inactive to active. 
Additionally, $\beta$-catenin can degrade independent of the destruction complex. Synthesis of $\beta$-catenin occurs at a constant rate.
\par
\textit{Wnt on:} When receptors on the surface of a cell bind to Wnt, the Wnt signaling transduction pathway is initiated.
This enables \bcat~to move into the nucleus where it binds with transcription factors that regulate genes. 
This signal propagation is mediated by the following molecular interactions. After Wnt stimulus,
 the protein Dishevelled is activated near the membrane. This in turn inactivates the destruction 
 complex, thereby preventing the destruction of $\beta$-catenin, allowing it to accumulate in the cytoplasm
  through natural synthesis. Throughout the molecular interactions in the signaling pathway, 
  intermediate complexes can form (e.g., \bcat~bound with Dishevelled). 
\par
\textit{Space:} The location of molecules plays a pivotal role: \bcat~moves between the 
cytoplasm and the nucleus (to reach target genes and regulate them). Dishevelled and molecules that form the destruction complex
shuttle between the nucleus and the cytoplasm. However, it is assumed that only the inactive destruction complex can shuttle (since in the cytoplasm it would be bound to $\beta$-catenin). 
Phosphatases exist in both the nucleus and the cytoplasm but the movement across compartments is not included in the model.  Symmetry of reactions is assumed if the species exist in both compartments.
Intermediate complexes are assumed to be short-lived, or not large enough for movement across compartments.

\par
The Wnt shuttle model of \cite{MRBH} has
$19$ species whose interactions can be framed as biochemical reactions.
 These species correspond to
variables $x_1, \ldots, x_{19}$ in our dynamical system
(\ref{eq:diffeqn}). Namely, $x_i$ represents the concentration of the species that is
listed in the $i$th row in Table~\ref{tab-notation}.

\begin{longtable}{|c|c|l|}
	\hline  {\bf Variable}  & {\bf Species} & {\bf Symbol} \\ \hline
		& {\bf Dishevelled} 		& $\boldsymbol{D}$  \\ 
	$x_1$ & Dishevelled in cytoplasm (inactive) 		& $D_i$  \\ 
	$x_2$ &  Dishevelled in cytoplasm (active) 		& $D_a$  \\ 
	$x_3$ &  Dishevelled in nucleus (active) 		& $D_{an}$   \\	\hline
	 & {\bf Destruction complex  (APC/Axin/GSK3$\beta$)}	& $\boldsymbol{Y}$	\\
	$x_4$& Destruction complex in cytoplasm (active)	&   $Y_a$	\\
	$x_5$  & Destruction complex in cytoplasm (inactive)	&   $Y_i$	\\
	$x_6$  & Destruction complex in nucleus (active)	&   $Y_{an}$	\\
	$x_7$  & Destruction complex in nucleus (inactive)	&   $Y_{in}$	\\ \hline
	  & {\bf Phosphatase}		& 	$\boldsymbol{P}$ \\ 
	 $x_8$ & Phosphatase in cytoplasm		& 	 $P$ \\ 
	 $x_9$ & Phosphatase in nucleus		&   $P_n$ 	\\ \hline
	&	 $\boldsymbol{\beta-}${\bf catenin} 			&   $\boldsymbol{x}$ 		\\
	$x_{10}$ & \bcat in cytoplasm			& 	$x$  	\\
	$x_{11}$ 	 & \bcat in nucleus			& 	  $x_n$ \\ \hline
	&		 { \bf Transcription Factor} 			&  $\boldsymbol{T}$	\\ 
	 $x_{12}$ & TCF (gene transcription in nucleus)			& $T$ 	\\ \hline
	 	& {\bf Intermediate complex}		& $\boldsymbol{C}$ \\ 
	$x_{13}$ 	& Transcription complex, $\beta$-catenin: TCF in nucleus		& 	$C_{xT}$ \\ 
	 $x_{14}$ & Intermediate complex, $\beta$-catenin: dishevelled in cytoplasm		&  $C_{YD}$ 	\\ 
	$x_{15}$ 	& Intermediate complex, destruction complex: dishevelled in nucleus		& $C_{YDn}$ 	\\ 
	 $x_{16}$ & Intermediate complex, destruction complex: phosphatase in cytoplasm		& $C_{YP}$ 	\\
	 $x_{17}$ & Intermediate complex, destruction complex: phosphatase in nucleus		& $C_{YPn}$  	\\
	 $x_{18}$ & Intermediate complex, $\beta$-catenin: destruction complex in cytoplasm	&$C_{xY}$  	\\ 
	 $x_{19}$ 	& Intermediate complex, $\beta$-catenin: destruction complex in nucleus		&  $C_{xYn}$  	\\ 

	\hline
\caption{The $19$ species in the Wnt shuttle model.}
\label{tab-notation}
\end{longtable}

The second column in Table~\ref{tab-notation}
indicates the biological meaning of the $19$ species.
The symbols in the last column are those used in
the presentation of the Wnt shuttle model in~\cite{MRBH}.

The $19$ species in the model interact according to the $31$ reactions given in Table~\ref{tab-reactions}.
Each reaction comes with a rate constant $k_i$. These are the 
coordinates of our parameter vector~${\bf k} $.

\medskip

\begin{longtable}{|c|l|}
	\hline {\bf Reaction} & {\bf Explanation}  \\ \hline
	 ${\xymatrix {x_{1}\ar@<0.3ex>[r]^(.5){k_{1}} & x_{2} \ar@<0.3ex>[l]^(.5){k_{2}}}}$	& (In)activation of dishevelled, depends on Wnt		 \\ 
	 ${\xymatrix{x_{2} + x_{4}  \ar@<0.3ex>[r]^(.6){k_3}  & x_{14} \ar@<0.3ex>[l]^(.4){k_4}  \ar[r]^(.4){k_5} & x_{2} + x_{5} }}$
							& Destruction complex active $\rightarrow$ inactive  \\
	${\xymatrix {x_{5} + x_{8} \ar@<0.3ex>[r]^(.6){k_6}  & x_{16} 	\ar@<0.3ex>[l]^(.4){k_{7}} \ar[r]^(.4){k_{8}} & x_{4} + x_{8}}}$
							& Destruction complex inactive $\rightarrow$ active  \\ 
		${\xymatrix {x_{4} + x_{10} 	\ar@<0.3ex>[r]^(.6){k_9}  & 	x_{18} 	\ar@<0.3ex>[l]^(.4){k_{10}} \ar[r]^(.4){k_{11}} &x_{4} + \emptyset}}$
	 						& Destruction complex-dependent \bcat degradation  \\
							${\xymatrix {\emptyset \ar[r]^(.5){k_{12}}& x_{10} }}$  	& \bcat production			 \\
 	${\xymatrix {x_{10}  \ar[r]^(.5){k_{13}} & \emptyset}}$	& Destruction complex-independent \bcat degradation	 \\ \hline	 
	   
	${\xymatrix {x_{3} + x_{6} \ar@<0.3ex>[r]^(.6){k_{14}}  & x_{15} 	\ar@<0.3ex>[l]^(.4){k_{15}} \ar[r]^(.4){k_{16}} & x_{3} + x_{7} }}$
	    						& Destruction complex active $\rightarrow$ inactive (nucleus) \\
	${\xymatrix {x_{7} + x_{9} \ar@<0.3ex>[r]^(.6){k_{17}}  & x_{17} \ar@<0.3ex>[l]^(.4){k_{18}} \ar[r]^(.4){k_{19}} & x_{6} + x_{9} }}$
							& Destruction complex inactive $\rightarrow$ active (nucleus) \\
		${\xymatrix {x_{6} + x_{11} \ar@<0.3ex>[r]^(.6){k_{20}}  &	x_{19} \ar@<0.3ex>[l]^(.4){k_{21}} \ar[r]^(.4){k_{22}} & x_{6} + \emptyset}}$												& Destruction complex-dependent \bcat degradation (nucleus)  \\
							${\xymatrix {x_{11} \ar[r]^(.5){k_{23}}& \emptyset}}$ & Destruction complex-independent \bcat degradation (nucleus)		\\ 
	${\xymatrix {x_{11} + x_{12} \ar@<0.3ex>[r]^(.6){k_{24}}  & x_{13} \ar@<0.3ex>[l]^(.4){k_{25}}}}$ & \bcat binding to TCF (nucleus)   \\ \hline
		   	${\xymatrix {x_{2} \ar@<0.3ex>[r]^(.5){k_{26}} & x_{3} \ar@<0.3ex>[l]^(.5){k_{27}}}}$ & Shuttling of active dishevelled   \\
	${\xymatrix {x_{5} \ar@<0.3ex>[r]^(.5){k_{28}} & x_{7} \ar@<0.3ex>[l]^(.5){k_{29}}}}$ & Shuttling of inactive-form destruction complex   \\
	   	${\xymatrix {x_{10} \ar@<0.3ex>[r]^(.5){k_{30}} & x_{11} \ar@<0.3ex>[l]^(.5){k_{31}}}}$		& Shuttling of \bcat  \\ \hline
\caption{The $31$ reactions in the Wnt shuttle model.}
\label{tab-reactions}
\end{longtable}

\vspace{-0.2in}

The $31$ reactions in Table \ref{tab-reactions}
translate into a dynamical system 
$\, {\dot{{\bf x}}} = \Psi({\bf x}; {\bf k})$.
Here $\Psi$ is a vector-valued function of the
vectors of species concentrations ${\bf x}$ and
rate constants ${\bf k}$. The choice of $\Psi$ is up to the modeler.
In this paper, we assume that $\Psi$
represents the {\em law of mass action} \cite[\S 2.1.1]{Klipp}.
This is precisely what is used in \cite{MRBH} for the
Wnt shuttle model. The resulting
dynamical system is  (\ref{eq:diffeqn}).
We refer to \cite{CF,FW,KPDDG,PDSC,ShSt} and their
many references for mass action kinetics and its variants.
In summary, Table \ref{tab-reactions} translates into
the dynamical system (\ref{eq:diffeqn}) under the law of mass action.
The five relations in (\ref{eq:conservation}) constitute a basis for the
linear space of conservation relations of the model in Table \ref{tab-reactions}
assuming mass action kinetics.

We refer to $x_1, \ldots, x_{19}$ as the \emph{species concentrations}, 
$k_1, \ldots, k_{31}$ as the \emph{rate parameters}, and
 $c_1, \ldots, c_5$ as the \emph{conserved quantities}.  
We write {\bf x}, {\bf k} and {\bf c} for the vectors with these coordinates.
As is customary in algebraic geometry, we take the coordinates 
in the complex numbers $ \mathbb{C}$, or possibly in some
other algebraically closed field $K$ containing the rationals~$\mathbb{Q}$.

Our aim is to understand the relationships between ${\bf x}, {\bf k}$ and ${\bf c}$ in the
 Wnt shuttle model. To this end, we introduce the \emph{steady state variety} $\mathcal S \subset \mathbb C^{55}$.
 This is the set of all points $({\bf x},{\bf k},{\bf c})$ that satisfy the equations
   ${\dot{x}_{1}}=\ldots={\dot{x}_{19}}=0$ in (\ref{eq:diffeqn})
   along with the     five conservation laws in (\ref{eq:conservation}).
   We write our ambient affine space as
 $\mathbb C^{55} = \mathbb C_{\bf x}^{19} \times \mathbb C_{\bf k}^ {31} \times \mathbb C_{\bf c}^5$.
 This emphasizes
  the distinction between the species concentrations, rate parameters, and conserved quantities.    
   
\section{Ideals, Varieties, and Nine Points}
\label{sec3}

We write $I$ for the ideal in the polynomial ring
$\mathbb{Q}[{\bf x},{\bf k}] = 
\mathbb{Q} [x_1, \ldots x_{19}, k_1, \ldots k_{31}]$
that is generated by the $19$ polynomials $\dot{x}_i$ on the right hand side of 
(\ref{eq:diffeqn}).
Five of these generators are redundant. Indeed, the 
conservation relations (\ref{eq:conservation}) give the
following identities modulo $I$:
$$
\begin{matrix}
	 \dot{x}_1+\dot{x}_2+\dot{x}_3+\dot{x}_{14}+\dot{x}_{15} \,= \,
	  \dot{x}_8+\dot{x}_{16} \,= \,
	 \dot{x}_9+\dot{x}_{17} \,= 
	 \dot{x}_{12}+\dot{x}_{13}\, = \\
	 	 \dot{x}_4+\dot{x}_5+\dot{x}_6+\dot{x}_7+\dot{x}_{14}
		 +\dot{x}_{15}+\dot{x}_{16}+\dot{x}_{17}+\dot{x}_{18} + \dot{x}_{19} \,\,=\,\, 0 .	 \end{matrix}
$$
For instance, the polynomials $\dot{x}_{13}, \dot{x}_{15}, \dot{x}_{16},\dot{x}_{17}$ and $\dot{x}_{19}$
are redundant
because they can be expressed as negated sums of other generators of $I$.
Hence $I$ is generated by $14$ polynomials. The variety $V(I)$ lives in the $50$-dimensional affine space
$\mathbb C_{\bf x}^{19} \times \mathbb C_{\bf k}^ {31}$,
and it is isomorphic to the steady state variety $\mathcal{S} \subset \CC^{55}$.
A direct computation using the computer algebra package
{\tt Macaulay2} \cite{M2} shows that $V(I)$ has dimension $36$.
Hence the affine ideal $I$ is a complete intersection in $\mathbb{Q}[{\bf x},{\bf k}]$.
Furthermore, using {\tt Macaulay2} we can verify the following lemma.

\begin{lemma}
\label{lem:3_1}
The ideal $I$ admits the non-trivial decomposition $I = I_m  \cap  I_e$,
where $I_e = I:\langle x_1 \rangle$ and $I_m = I+\langle x_1 \rangle$,
both of these components have codimension $14$, and $I_e$ is a prime ideal.
\end{lemma}

The ideal $I_m$ is called the {\em main component}, while $I_e$ is called the {\em extinction component}, since it reflects those steady states where a number of the reactants ``run out."  
Both of these ideals live in $\mathbb{Q}[{\bf x},{\bf k}]$, 
and we now present explicit generators. 
The extinction component equals
$$ \begin{matrix}
I_e = \langle x_1,x_2,x_3,x_5,x_7,x_{14},x_{15},x_{16},x_{17},  
k_{30}x_{10}-(k_{23}+k_{31})x_{11}-k_{22}x_{19}, & \\
k_{13}x_{10}+k_{23}x_{11}+k_{11}x_{18}+k_{22}x_{19}-k_{12},
 k_{24}x_{11}x_{12}-k_{25}x_{13}, & \\
k_{20}x_{6}x_{11} - (k_{21} + k_{22})x_{19},
k_{9}x_{4}x_{10} - (k_{10} + k_{11})x_{18}
\rangle. &
\end{matrix} $$
The ideal  $I_e$ is found to be prime in $\mathbb{Q}[{\bf x},{\bf k}]$.
The main component equals
$$ \begin{matrix}
I_m \,\,=\,\, \langle k_{16}x_{15}-k_{19}x_{17},  
k_{5}x_{14}-k_{8}x_{16},  
k_{30}x_{10}-(k_{23}+k_{31})x_{11}-k_{22}x_{19},\qquad & \\ 
k_{13}x_{10}+k_{23}x_{11} + 
     k_{11}x_{18}+k_{22}x_{19}-k_{12},
     k_{28}x_{5}-k_{29}x_{7}, 
     k_{26}x_{2}-k_{27}x_{3}, & \\
     k_{1}x_{1}-k_{2}x_{2},  
     k_{24}x_{11}x_{12}-k_{25}x_{13},  
      k_{20}x_{6}x_{11}-(k_{21}+k_{22})x_{19}, & \\
     k_{9}x_{4}x_{10}-(k_{10}+k_{11})x_{18},
     k_{17}x_{7}x_{9}-(k_{18}+k_{19})x_{17},
     k_{6}x_{5}
     x_{8}-(k_{7}+k_{8})x_{16}, &\\
     k_{14}x_{3}x_{6}-k_{15}x_{15}-k_{19}x_{17}, 
     k_{3}x_{2}x_{4}-k_{4}x_{14}-k_{8}x_{16}, & \\ 
     (k_{4}k_{6}k_{8}k_{
     14}k_{16}k_{18}k_{26}k_{29}  
     +k_{5}k_{6}k_{8}k_{14}k_{16}k_{18}k_{26}k_{29}+ & \\
     k_{4}k_{6}k_{8}k_{14}k_{16}k_{19}k_{26}k_{
     29}+
     k_{5}k_{6}k_{8}k_{14}k_{16}k_{19}k_{26}k_{29})k_{1}x_{6}x_{8} & \\
     -(k_{3}k_{5}k_{7}k_{15}k_{17}k_{19}k_{27}k_{28}+
     k_{
     3}k_{5}k_{8}k_{15}k_{17}k_{19}k_{27}k_{28} &\\  +k_{3}k_{5}k_{7}k_{16}k_{17}k_{19}k_{27}k_{28}+  k_{3}k_{5}k_{8}k_{16}k_{17}
     k_{19}k_{27}k_{28})k_{1}x_{4}x_{9} \rangle .&
     \end{matrix}
     $$
This ideal is not prime in $\mathbb{Q}[{\bf x},{\bf k}]$.
For instance,  the variable $k_1$ is a zerodivisor modulo $I_m$,
as seen from the 
last generator. Removing the factor $k_1$ from the
last generator yields the quotient ideal $I_m:\langle k_1 \rangle$. However, even that ideal
 still has several associated primes. 
All of these prime ideals, except for one, contain some of the rate constants $k_i$.

That special component is characterized in the following proposition.
Given any ideal $J \subset \mathbb{Q}[\mathbf{x},\mathbf{k}]$, we write
$\widetilde{J} =  \mathbb{Q}(\mathbf{k}) [{\bf x}] J$ for its extension
to the polynomial ring $\mathbb{Q}(\mathbf{k}) [{\bf x}]$ in the unknowns $x_1,\ldots,x_{19}$ over the field
of rational functions in the parameters $k_1,\ldots,k_{31}$.

\begin{prop} 
\label{prop:3_2} The ideal $J_m = \widetilde{I_m} \,\cap\, \mathbb{Q}[{\bf x},{\bf k}]$
is prime. Its irreducible variety $V(J_m) \subset \CC^{50}$ has dimension $36$; it is the unique
component of $V(I_m)$ that maps dominantly onto $\CC_{\bf k}^{31}$.
\end{prop}

\begin{proof}
The ideal $\widetilde{I_m}$ has the same generators as $I_m$
but now regarded as polynomials in ${\bf x}$ with coefficients in 
$\mathbb{Q}(\mathbf{k})$. Symbolic computation in the  ring
$\mathbb{Q}(\mathbf{k}) [{\bf x}] $ reveals that $\widetilde{I_m}$
is a prime ideal. This implies that $J_m$ is a prime ideal in
$\mathbb{Q}[{\bf x},{\bf k}]$, and hence  $V(J_m)$ is irreducible. The dimension statement follows from
the result of Lemma \ref{lem:3_1} that $I_m$ is a complete intersection.
This ensures that $V(I_m)$ has no lower-dimensional components, by Krull's Principal Ideal Theorem.
Finally, $V(J_m)$ maps dominantly onto 
$\CC_{\bf k}^{31}$
because $J_m\,\cap\, \mathbb{Q}[\mathbf{k} ] = \{0\}$.
\end{proof}

\begin{cor}
The ideal $\widetilde{I}$ is radical, and it is the intersection of two primes in
$\mathbb{Q}(\mathbf{k}) [{\bf x}] $:
\begin{equation}
\label{eq:twocomponents}
 \widetilde{I} \,\, = \,\,  \widetilde{I_e} \,\cap\,\widetilde{I_m}.
\end{equation}
\end{cor}

\begin{proof}
This follows directly from Proposition \ref{prop:3_2}
and the primality of $I_e$ in Lemma \ref{lem:3_1}.
\end{proof}

The decomposition has the following geometric interpretation.
We now work over the field $K = \overline{\mathbb{Q}({\bf k})}$.
All rate constants are taken to be generic. Then
$V(\widetilde{I})$ is the
$5$-dimensional variety of all steady states
in $K^{19}$. This variety is the union of two irreducible
components,
$$ V(\widetilde{I}) \,\, = \,\, V(\widetilde{I_e})\,\cup\, V(\widetilde{I_m}), $$
where each component is $5$-dimensional. The first component
lies inside the $10$-dimensional coordinate subspace
$V(x_1,x_2,x_3,x_5,x_7,x_{14},x_{15},x_{16},x_{17})$. 
Hence it is disjoint from the hyperplane defined by the
first conservation relation $\,x_1 + x_2 + x_3 + x_{14} + x_{15} = c_1$.
In other words, $V(\widetilde{I_e})$ is mapped into a coordinate hyperplane
under the map $\chi : K^{19} \rightarrow K^5, {\bf x} \mapsto {\bf c}$.

On the other hand, the second component $V(\widetilde{I_m})$ maps dominantly
onto $K^5$ under $\chi$. Theorem  \ref{thm:nine} states that
the generic fiber of this map consists of $9$ reduced points.
Equivalently,
\begin{equation}
\label{eq:chichi}
 \chi^{-1}({\bf c}) \cap V(\widetilde{I}) \,= \,
\chi^{-1}({\bf c}) \cap V(\widetilde{I_m}) 
\end{equation}
is a set of nine points in $K^{19}$.
We are now prepared to argue that this is indeed the case.

\begin{proof}[Computational Proof of Theorem \ref{thm:nine}] 
We consider the ideal of the variety (\ref{eq:chichi}) in the polynomial ring
$\mathbb{Q}({\bf k},{\bf c})[{\bf x}]$. This polynomial ring has $19$ variables,
and all $36$ parameters are now scalars in the coefficient field. This ideal is generated by
the right hand sides of (\ref{eq:diffeqn}) and (\ref{eq:conservation}).
Performing a Gr\"obner basis computation in this polynomial ring
verifies that our ideal is zero-dimensional and has length $9$. Hence
(\ref{eq:chichi}) is a reduced affine scheme of length $9$ in $K^{19}$.

Fast numerical verification of this result is obtained by
replacing the coordinates of ${\bf k}$ and~${\bf c}$ with generic (random rational) values.
In {\tt Macaulay2} one finds, with probability~$1$,
that the resulting ideals in $\mathbb{Q}[{\bf x]}$ are radical of length $9$.
   We also verified this result via {\em numerical algebraic geometry}, using
    the two software packages  {\tt Bertini} \cite{BHSW} and {\tt PHCpack} \cite{Ver}.
\end{proof}

\section{Multistationarity and its Discriminant}
\label{sec4}

This section centers around Question 4 from the Introduction:
{\em For what real positive rate parameters and conserved quantities 
does the system exhibit multistationarity?} 
This is commonly asked about biochemical reaction networks and about dynamical systems in general.  

Mathematically, this is a problem of {\em real} algebraic geometry.
Writing $\mathcal{S}$ for the steady state variety in $\CC^{55}$,
 we are interested in the fibers of the map
$\pi_{{\bf k},{\bf c}} : \mathcal{S} \, \cap \, \RR_{>0}^{55} \rightarrow 
\RR_{>0,{\bf k}}^{31} \times \RR_{>0,{\bf c}}^5$.
According to  Theorem \ref{thm:nine}, the general fiber
consists of $9$ {\em complex} points ${\bf x} \in \CC_{\bf x}^{19}$,
when the map  $\pi_{{\bf k},{\bf c}} $ is taken over $\CC$. But here we take it 
over the reals $\RR$ or over the positive reals $\RR_{> 0}$.

In our application to  biology, we only care about concentration vectors
${\bf x} $ whose coordinates are real and positive. Thus we wish to stratify
$\RR_{>0,{\bf k}}^{31} \times \RR_{>0,{\bf c}}^5$ according to the cardinality of
\begin{equation}
\label{eq:pipi}
 \pi^{-1}_{{\bf k},{\bf c}} ({\bf k},{\bf c}) \, = \, \bigl\{ \,
 ({\bf x}, {\bf k}', {\bf c}') \in  \mathcal{S} \, \cap \, \RR_{>0}^{55} \ : \  \
  {\bf k}' = {\bf k} \,\,\hbox{and} \,\, {\bf c}' = {\bf c}
  \bigr\}. 
 \end{equation}
This stratification comes from a decomposition of the $36$-dimensional orthant
$\RR_{>0,{\bf k}}^{31} \times \RR_{>0,{\bf c}}^5$ into
connected open semialgebraic subsets. The walls in this decomposition are 
given by the {\em discriminant} $\Delta$,  a giant polynomial in
the $36$ unknowns $({\bf k},{\bf c})$ that is to be defined later.

We begin with the following result on what is possible with regard to real positive solutions.

\begin{thm}\label{thm:posrealsols}
Consider the polynomial system in
(\ref{eq:diffeqn})--(\ref{eq:conservation}) 
where all parameters $k_i$ and $c_j$ are positive real numbers.
The set (\ref{eq:pipi}) of positive real solutions can have $1,2$, or $3$ elements.
\end{thm}

\begin{proof}
For random choices of  $({\bf k},{\bf c}) = (k_1, \ldots, k_{31}, c_1, \ldots, c_5)$ in  the orthant $\RR_{>0}^{36}$,
our polynomial system has $9$ complex solutions, by Theorem \ref{thm:nine}.
For the following  two special choices of the $36$ parameter values, all $9$ solutions are real.
First, take $({\bf k},{\bf c})$ to be the vector
\begin{small}
$$ \begin{matrix}
&( 1.7182818, 53.2659, 3.4134082, 0.61409879, 0.61409879, 3.4134082, 
  0.98168436,0.98168436,  \\ & 
  92.331732, 0.86466471, 79.9512906, 97.932525, 1,
  3.2654672, 0.61699064, 0.61699064,\\ &  37.913879,  0.86466471, 0.86466471, 
  4.7267833, 0.17182818, 0.68292191,  1, 0.55950727, \\ &  1.0117639, 1.7182818, 
  1.7182818, 0.99326205, 0.99326205, 5.9744464, 1, 4.9951026,  \\ & 16.4733784, 
1.6006340000000001, 1.2089126, 2.7756596399999998).
  \end{matrix}
  $$
\end{small}
  The resulting system has three positive solutions ${\bf x} \in \RR_{>0}^{19}$. Next, 
  let $({\bf k}',{\bf c}')$ be the vector 
  \begin{small}
  $$ \begin{matrix}
&(0.948166, 7.45086, 5.72974, 3.96947, 7.21145, 7.8761, 1.87614,  8.11372, 6.21862, 5.24801, \\ &
3.10707,  1.08146, 5.22133, 5.84158, .911392,  4.28788, 4.81201, 9.67849, 1.34452, 7.38597,
\\ & 6.64451, 7.10229, 8.57942, 
5.79076, 6.33244, 1.53916, 1.39658, 0.81673, 5.8434, 3.86223,  \\ & 
7.22696, \,\,1.45438, 3.36482, 6.06453, 4.82045, 3.6014).
 \end{matrix}
 $$
 \end{small}Here, one solution to our system is positive.
 By connecting the two parameter points above with a general curve in $\RR_{>0}^{36}$,
 and by examining in-between points $({\bf k}'',{\bf c}'')$, we can construct a system
 with two positive solutions.
 All computations were carried out using {\tt Bertini}  \cite{BHSW}.
\end{proof}

\begin{remark}
At present, we do not know whether the number of real positive solutions
can be larger than three. We suspect that this is impossible, but we currently cannot prove it.
\end{remark}

The difficulty lies in the fact that the stratification of
$\RR_{> 0}^{36}$ is extremely complicated. In computer algebra,
the derivation of such stratifications
is known as the problem of {\em real root classification}.
For a sample of recent studies in this direction see \cite{CDMXX, FMRS, RT}.
Real root classification is challenging
even when the number of parameters is $3$ or $4$; clearly, $36$ parameters is out of the question. The stratification
of $\RR_{>0}^{36}$ by behavior of (\ref{eq:pipi})
has way too many cells.

While symbolic techniques for real root classification are infeasible
for our system, we can use numerical algebraic geometry \cite{gross15}
 to gain insight into the stratification of $\RR_{> 0}^{36}$. 
   \emph{Coefficient-parameter homotopies} \cite{MS} can solve the 
steady state polynomial system (\ref{eq:diffeqn})-(\ref{eq:conservation})
   for multiple choices of $({\bf k}, {\bf c})$ quickly. 
  For our computations we use {\tt Bertini.m2}.
  This is the {\tt Bertini} interface for {\tt Macaulay2}, 
as described in \cite{BGLR}.
 Each system has $19$ equations in $19$ unknowns and,
 for random $({\bf k},{\bf c})$, each system has $9$ complex solutions.
 Such a system can be solved in less than one second using
the {\tt bertiniParameterHomotopy} function from {\tt Bertini.m2}.

Below we describe the following experiment.
We sample $10,000$ parameter vectors $({\bf k},{\bf c})$ 
from two different probability distributions on $\mathbb{R}_{> 0}^{36}$.
In each case we report  the observed
frequencies for the number of real solutions and 
 number of positive solutions.  We then follow these experiments with a specialized sampling scheme
 for testing numerical robustness.

\emph{Uniform sampling scheme:}
Here we choose $({\bf k}, {\bf c} )$ uniformly from
the cube $(0.0, 100.0)^{36}$.  Sampling 10,000 parameter vectors from this scheme and solving the steady state system for each of these parameter vectors in {\tt Bertini}, we obtained $9,992$ solutions sets that contained $9$
complex points. Solution sets with less than $9$ points occur when some paths in the coefficient-parameter homotopy fail. We call solution sets with $9$ solutions {\em good}.



\emph{Integer sampling scheme:}  Here we select $({\bf k}, {\bf c})$ uniformly from $\{1,2,3\} ^{36}$.  Sampling 10,000 parameter vectors according to this scheme and solving the corresponding steady state system
returned $9,963$ good solution sets.  Below is a table that records how many of the good solution sets 
had $9,7,5,3$ real solutions; all solution sets had 1 positive real solution.


\begin{longtable}{|l||c|c|c|c|}
\hline
\# of real solutions &			 9 & 	     7 		& 5 		& 3 	\\ \hline
Freq. for Uniform Sampling &  	5,760 & 3,675	& 544  	& 13\\
Freq. for Integer Sampling &  	 2,138&  5,181	& 2,522 	& 122  \\
\hline
\caption{Frequencies for the sampling schemes.}
\end{longtable}

\vspace{-6mm}

These computations indicate that for most parameter vectors in $(0, 100)^{36}$ we will see only one positive solution to the steady state system. But while the set of parameter vectors that result in multiple steady states is not very large, we can give evidence that multistationarity is preserved under small perturbations.
This is our next point.

\emph{Testing Robustness:} Let $({\bf k^*}, {\bf c^*})$ be the first point in the proof of Theorem \ref{thm:posrealsols}.  For each index $i \in \{1,\ldots,19\}$ we choose $y_i$ uniformly from $(-0.03\cdot k^*_i, 0.03\cdot k^*_i)$ then set $k_i=k_i^*+y_i$.  We ran the same process for the $c_i$. Sampling $10,000$ parameter vectors this way and solving the corresponding steady state systems returned $10,000$
good solution sets, as follows:

\begin{longtable}{|c|c|c|c|c|}
\hline
\# of real solutions & Freq. & & \# of pos. solutions & Freq. \\
\hline
9 & 9,879 & & 3 &9,879\\
7 & 121 & & 1 & 121\\
\hline
\caption{Frequencies for testing robustness scheme.}
\end{longtable}

In the remainder of this section, we 
properly define the discriminant $\Delta$ that separates  the various 
strata in $\mathbb{R}^{36}_{>0}$.
 Let $\Delta_{\rm int}$ denote the Zariski closure in
$ \mathbb C_{\bf k}^{31} \times \mathbb C_{\bf c}^ {5}$ 
of all parameter vectors $({\bf k},{\bf c})$ for which
 (\ref{eq:diffeqn})--(\ref{eq:conservation})  does not have
$9$ isolated complex solutions and there are no solutions
with $x_i = 0$ for some $i$. It can be shown that 
$\Delta_{\rm int}$ is a hypersurface that is defined over $\mathbb{Q}$,
so it is given by a unique (up to sign) irreducible
squarefree polynomial in $\mathbb{Z}[{\bf k},{\bf c}]$.
We use the symbol $\Delta_{\rm int}$ also for that polynomial.
To be precise, $\Delta_{\rm int}$ is the discriminant
of a number field $L$ 
with $K \supset L \supset \mathbb{Q}$, namely $L$ is
the field of definition of the finite $K$-scheme~(\ref{eq:chichi}).

Next, for any $i \in \{1,2,\ldots,19\}$ consider the
intersection of the steady state variety $\mathcal{S}$
with the hyperplane $\{x_i = 0\}$.
The Zariski closure of the image of $\mathcal{S} \cap \{x_i = 0\}$
under the map $\pi_{{\bf k},{\bf c}}$ is a hypersurface
in $ \mathbb C_{\bf k}^{19} \times \mathbb C_{\bf c}^ {31}$,
defined over $\mathbb{Q}$, and we write $\Delta_{x_i = 0}$ for the unique (up to sign)
irreducible polynomial in $\mathbb{Z}[{\bf k},{\bf c}]$ that vanishes on that hypersurface. We now define
$$ \Delta \,\,\,:= \,\,\, \Delta_{\rm int} \cdot {\rm lcm} 
\bigl(\,\Delta_{x_1 = 0}\,,\,\Delta_{x_2 = 0}\,,\,\ldots\,,\,\Delta_{x_{19} = 0}\,\bigr).$$
This product with a least common multiple (lcm) is the  {\em discriminant} for our problem.

\begin{ex}
The degree of $\Delta_{\rm int}$ as a polynomial
only in ${\bf c} = (c_1,c_2,c_3,c_4,c_5)$ equals $34$.
To illustrate this, we set ${\bf c} = \bigl(5,16 + C,\frac{8}{5}-C, \frac{6}{5}+C, 3-C \bigr)$
      where $C$ is a parameter, and
$$ {\bf k} \,=\,  \biggl(\frac{9}{5}, \frac{9}{5}, 3, \frac{2}{3}, \frac{2}{3}, 3, 1, 1, 100, 
\frac{4}{5}, 80, 100, 1, 3, \frac{2}{3}, \frac{2}{3}, 38, \frac{4}{5}, \frac{4}{5}, 4, 
\frac{1}{8}, \frac{3}{5}, 1, \frac{1}{2}, 19, \frac{7}{4}, \frac{7}{4},  1, 1, 5, 1\biggr).$$
 Under this specialization, the polynomial $\Delta_{\rm int}$ becomes 
 an irreducible polynomial of degree $34$ in the parameter $C$.
 Its coefficients are enormously large integers. It has $14$ real roots.

For the other factors $\Delta_{x_i = 0}$ of the discriminant, we find the following specializations:
\begin{equation}
\label{eq:otherdisc}
\begin{small}
 \begin{matrix}
x_1 \rightarrow 0,\,\,
x_2 \rightarrow 0,\,\,
x_3 \rightarrow 0,\,\,
x_4 \rightarrow (C{+}16)(5C{-}8),\,\,
x_5 \rightarrow C{+}16,\,\,
x_6 \rightarrow (C{+}16)(5C{+}6),\, \\
x_7 \rightarrow C{+}16,  
x_8 \rightarrow 5C-8,\,
x_9 \rightarrow 5C+6,\,
x_{10} \rightarrow \hbox{a quartic}\,q(C),\,
x_{11} \rightarrow 0,\,
x_{12} \rightarrow C{-}3,\,\\
x_{13} \rightarrow C{-}3,  \,\,
x_{14} \rightarrow (C{+}16)(5C{-}8), \,\,
x_{15} \rightarrow (C{+}16)(5C{+}6),\,\,
x_{16} \rightarrow (C{+}16)(5C{-}8), \\
x_{17} \rightarrow (C{+}16)(5C{+}6), \,\,\,
x_{18} \rightarrow (C{+}16)(5C{-}8)q(C),\,\,\,
x_{19} \rightarrow (C{+}16)(5C{+}6).
\end{matrix}
\end{small}
\end{equation}
These polynomials have $8$ distinct real roots in total, so
the total number of real roots 
of the discriminant  is $14+8 = 22$. These
are the break points where real root behavior changes:
$$ \begin{matrix}
     (9,0) & {\bf -77.2388} & (9,0) & {\bf -16.0000} & (9,0)  & -5.28669 &    (7,0) &
     -1.57472 \\  (9,0) &  {\bf -1.46506} & (9,0) & -1.34899  & (7,0) & -1.29581 &
      (9, 0) &    {\bf -1.20000} \\ (9,1)   & {\bf -1.19215} & (9,1) & -1.18389 & (7,1) &
       -0.584325  & (9,3) & -0.361808 \\   (7,3) & 0.191039  & (5,1) & 1.30812 & (7,1) &  1.33197 &
        (5,1)  &  {\bf 1.60000}  \\
   (5,0) & 1.60161 &(3,0) & {\bf 3.0000} & (3,0)  &  4.26306 & (5,0) & 11.1174 \\
   (7,0) & 21.4165 & (9,0) & {\bf 310.141} & (9,0) 
\end{matrix}
$$
In this table, we list all $22$ roots of the specialized discriminant $\Delta(C)$.
The eight boldface values of $C$ are the roots of (\ref{eq:otherdisc}):
here one of the coordinates of ${\bf x}$ becomes zero. At the other
$14$ values of $C$, the number of real roots changes. Between
any two roots we list the pair $(r,p)$, where $r$ is the number of real roots
and $p$ is the number of positive real roots. For instance, for
$ -0.361808 <C < 0.191039$, there are $7$ real roots of which $3$  are positive.
\end{ex}

\section{Algebraic Matroids and Parametrizations}
\label{sec5}
Question 5 asks: {\em Suppose we can measure only a subset of the 
species concentrations. Which subsets can lead to model rejection?}
This issue is important for the Wnt shuttle model
 because, in the laboratory, only some of the species are measurable by existing techniques. 

We shall address Question 5 using {\em algebraic matroids}.  
Matroid theory allows us to analyze the structure of relationships 
 among the $19$ species in Table \ref{tab-notation}. This first appeared in \cite{MRBH}.
  We here present an in-depth study of the matroids 
 that govern the Wnt shuttle model.
 
An introduction to (algebraic) matroids can be found in \cite{Ox}; they have been 
applied in \cite{KTT, KRT} to problems involving the completion of partial information.
General algorithms for computing algebraic matroids are derived in \cite{Ros}.
We briefly review  basic notions.

\begin{defn}
A {\em matroid} is an ordered pair $(X,\mathcal{I})$, where $X$ is a finite set,
here regarded as unknowns, and
$\mathcal{I}$ is a subset of the power set of $X$. These satisfy certain
{\em independence axioms}. For an {\em algebraic matroid}, we are given a prime ideal $P$
in the polynomial ring $K[X]$ generated by $X$, and $\mathcal{I}$ consists of
 subsets of $X$ whose images in $K[X]/P$ are  algebraically independent over $K$.
Thus, the collection of independent sets is
$\,\mathcal{I} \,= \,\bigl\{Y \subseteq X \,:\,P \cap K[Y] = \{0\} \bigr\}$.
\begin{enumerate}[1.]
\item {\em Bases} are maximal independent sets, i.e.~subsets in $\mathcal{I}$ that have maximal cardinality.
\item {\em Rank} is a function $\rho$ from the power set of $X$ to the natural numbers, which takes as input a set $Y \subset X$ and returns the cardinality of the largest subset of $Y$ in $\mathcal{I}$.
\item {\em Closure} is a function from the power set $2^X$ to itself. The input is a set $Y$ and the output is the largest set containing $Y$ with the same rank.
\item {\em Flats} are the elements  in $2^X$ that lie in the image of the closure map.
\item {\em Circuits} are the sets of minimal cardinality {\bf not} contained in $\mathcal{I}$.
\end{enumerate}
\end{defn}

We are here interested in the matroid that is defined by the prime ideal
$P = \widetilde{I_m}$ in $\mathbb{Q}({\bf k})[{\bf x]}$. Its ground set
$X$ is the set of species concentrations $\{x_1,\ldots,x_{19}\}$. Since
$V(\widetilde{I_m})$ is $5$-dimensional, each basis consists of
five elements in $X$. In our application, bases
are the maximal subsets of $X$ that can be specified independently at steady state; they 
 are also the minimal-cardinality sets that can be measured to learn all
 species concentrations.  The rank of a set $Y$ indicates the number of measurements required to learn the concentrations for every element of $Y$.  Flats are the full subsets that are specified 
  by any given collection of measurements.

Circuits furnish our answer to Question 5: they are minimal sets of species  that can be used to test compatibility of the data with the model. For each circuit $Y$ there is a unique-up-to-scalars
relation in $ \widetilde{I_m} \cap \mathbb{\bf Q}({\bf k})[Y]$, called the {\em circuit polynomial} of $Y$.
 If the measurements indicate that this relation is not satisfied, then the model and data are not compatible.

\begin{prop} \label{prop:951}
The algebraic matroid of $\widetilde{I_m}$ has rank $5$.
It has $ 951$ circuits, summarized in Table \ref{tab-circuits}.
Of the $11628$ subsets of $X$ of size $5$, precisely $2389$ are bases.
The $2092$ bases summarized in Table~\ref{tab-bases}  have base degree $1$, 
while the remaining $297$ have base degree $2$. 
\end{prop}

The computation of this matroid was carried out using the
methods described in \cite{Ros}. It was first reported in
\cite{MRBH}, along with the matroids of alternative models for the Wnt pathway. 
The idea there was to find subsets of variables that were dependent for different models.

 Our matroid analysis here goes beyond \cite{MRBH} in several ways:
\begin{enumerate}[1.]
\item We keep track of the parameters ${\bf k}$. We take our
circuit polynomials to have (relatively prime) coefficients in $\mathbb{Z}[{\bf k}]$.
This gives us a new tool for model rejection, e.g.~in situations 
 where only one data point is known but some parameter values are available.
\item We show how circuits can be used in parameter estimation; 
this will be done in Section~\ref{sec8}.
\item We use the degree-1 bases to derive rational parametrizations of the 
variety $V(\widetilde{I_m})$.
\end{enumerate}

We now explain  Table \ref{tab-circuits}.
A circuit polynomial has {\em type} $(i,j)$ if it contains $i$ 
species concentrations (${\bf x}$-variables) and $j$ rate parameters (${\bf k}$-variables).
The entry in row $i$ and column $j$ in Table \ref{tab-circuits}
is the number of circuits of type $(i,j)$.
Zero values are omitted for clarity.

\begin{table}[h]
\[
\begin{array}{|c|cccccc}\hline & 2 & 3 & 4 & 5 &  6  \\ \hline
       2  & 5 & 1 &   &  &    \\
       3  &  & 6 &   &  &    \\
       4  &1 &5 & &  &   \\
       5  & &6 &1 &  &   \\
       6  & &7 &5 &  &   \\
       7  & &5 &3 &  &   \\
       8  & &1 &11 &1 &    \\
       9  & &6 &12 &3 &   \\
       10 & & &11 &1 &    \\
       11 & &4 &7  &11 & 1  \\ \hline
              \end{array}
\begin{array}{|c|cccccc} \hline
       & 2 & 3 & 4 & 5 &  6  \\ \hline
       12 & & &13 &10 &  \\
       13 & & &13 &15& 2  \\
       14 & & &19 &16& 1  \\
       15 & & &17 &21& 4  \\
       16 & & &15 &11& 2  \\
       17 & & &16 &32& 9  \\
       18 & & &4 & 6 & 2  \\
       19 & & &26 &36 &11 \\ 
       20 & & &44 &1 & 1  \\
       21 & & &26 &27& 9  \\ \hline
                     \end{array}
\begin{array}{|c|ccccc|} \hline
       & 2 & 3 & 4 & 5 &  6  \\ \hline
       22 & & &8  &58&   \\
       23 & & &4  &56& 5  \\
       24 & & &  &54 &14 \\
       25 & & &  &53 &15 \\
       26 & & & & 8  &16 \\
       27 & & &12 &56 &16 \\
       28 & & &2  &  &2  \\
       29 & & &  &29 &14 \\
       30 & & &  & &   \\
       31 & & & &  & 6  \\ \hline
       \end{array}
       \]
       \caption{
       The $951$ circuit polynomials, by
       numbers of unknowns $x_i$ and $k_j$.}
\label{tab-circuits}
\end{table}       

\begin{ex} 
\label{ex:circs}
There are five circuits of type $(2,2)$. One of them is
$\,  \dot{x}_1 =  -k_{1} x_{1} + k_{2} x_{2}$.
Most of the $951$ circuit polynomials  in $\widetilde{I_m}$
are more complicated.
In particular, they are non-linear in both ${\bf x}$ and ${\bf y}$.
For instance, the unique circuit polynomial of type $(6,11)$ equals
$$ \begin{matrix}
 (-k_{15}k_{17}k_{19}k_{20}k_{25}-k_{16}k_{17}k_{19}k_{20}k_{25})x_7x_9x_{13} \\ +(k_{14}k_{16}k_{18}k_{21}k_{24}+k_{14}k_{16}k_{19}k_{21}k_{24}+k_{14}k_{16}k_{18}k_{22}k_{24}+k_{14}k_{16} k_{19}k_{22}k_{24})x_3x_{12}x_{19}. 
 \end{matrix}
 $$
In Section \ref{sec7}, we will consider the role of these nonlinear functions in parameter estimation.
\end{ex}

Given a basis $Y$ of an algebraic matroid, its {\em base degree} 
is the length of the generic fiber of the projection of  
$V(P)$ onto the $Y$-coordinates (cf.~\cite{Ros}). Bases with degree $1$ are desirable:


\begin{prop} \label{param}
Let $P \subset K[X]$ be a prime ideal, $Y$ a basis of its algebraic matroid,
$|X| =n$, and $|Y|= r$. If $Y$ has base degree $1$ then $V(P)$ is a rational variety, and 
the basic circuits of $Y$ specify a
birational map $\varphi_Y : K^r \dashrightarrow K^n$ whose  image is Zariski dense in
$V(P)$
\end{prop}

\begin{proof}
For each coordinate $x_i $ in $X \backslash Y$ there exists
a circuit containing $Y \cup \{x_i\}$; this is the {\em basic circuit} of $(Y,x_i)$.
Since $Y$ has base degree $1$, the generic fiber of the map $V(P) \rightarrow K^r$
consists of a unique point.
Therefore the circuit polynomial is linear in $x_i$. It has the form
\[ p_i(Y) \cdot x_i + q_i(Y) , \qquad \hbox{where} \,\, p_i,q_i \in K[Y].\]
The $i$-coordinate of the  rational map $\,\varphi_Y\,$ 
equals $\,x_i$ if $x_i \in Y\,$ and $\,-q_i(Y)/p_i(Y) $ if $x_i \notin Y$.
\end{proof}

From Propositions  \ref{prop:951} and \ref{param},
we obtain $2092$ rational parametrizations 
of the variety $V(\widetilde{I_m})$. These are the maps
$\varphi_Y : K^5 \dashrightarrow K^{19}$, 
where $Y$ runs over all bases of base degree $1$.
Using these $\varphi_Y$, we obtain $2092$
representations of the steady state variety 
(\ref{eq:chichi}) as a subset of $K^5$,
where now $K = \mathbb{Q}({\bf k},{\bf c})$.
Namely, we consider the preimages
of the five hyperplanes defined by (\ref{eq:conservation}).
These are hypersurfaces in $K^5$ whose intersection
represents the nine points in (\ref{eq:chichi}).
We performed the following computation for all $2092$ bases
$Y = \{y_1,\ldots,y_5\}$ of base degree~$1$:

\begin{enumerate}[1.]
\item Substitute ${\bf x} = \varphi_Y(y_1,\ldots,y_5)$ into the five linear equations (\ref{eq:conservation}).
\item Clear the denominators $d_1,\ldots,d_5$ in each equation to get polynomials $h_1,\ldots,h_5$ in $Y$.
\item The saturation ideal $J_Y = \langle h_1,\ldots,h_5 \rangle: \langle d_1 d_2 \cdots d_5\rangle^\infty$
represents the preimage of (\ref{eq:chichi}).
\end{enumerate}

Given such a wealth of parametrizations, we seek one where $J_Y$ has 
desirable properties. We use the following criterion:
consider subsets of five of the generators of $J_Y$,
compute the {\em mixed volume} of their Newton polytopes,
and fix a subset minimizing that mixed volume. 
In the census of $2092$ bases in  Table \ref{tab-bases},
that minimum is referred to as the mixed volume of $Y$.

\vspace{-7mm}
\begin{table}[h]
\[\begin{array}{|r|c|c|c|c|c|c|c|c|c|c|c|c|c|c|c|c|c|}\hline
\text{\em Mixed Volume} & 5& 9& 10& 11& 12& 13& 14& 15& 16& 20& 23& 24& 25& 30& 35& 42& 45 \\ \hline
\text{\em Frequency} & 2& 416& 6& 73& 50& 167& 563& 751& 10& 12& 6& 1& 11& 12& 4& 4& 4 \\ \hline
\end{array}\]
    \caption{Reducing the steady state equations to the  $2092$ bases of base degree $1$}
\label{tab-bases}
\end{table}
\vspace{-7mm}

By Bernstein's Theorem, the mixed volume is the number of solutions
to a generic system with the five given Newton polytopes.
We seek bases $Y$ where this matches the number nine from
Theorem~\ref{thm:nine}. We see that the mixed volume is nine
for $416$ of the bases in Table~\ref{tab-bases}.

\begin{ex}
 \label{ex:param}
The basis $Y = \{x_1, x_4, x_6, x_8, x_{13}\}$ has base degree $1$ and mixed volume $9$.
The remaining variables can be expressed in terms of $Y$ as follows. For brevity, we set
$$
\begin{matrix}
r(x_4,x_6) & = &
k_9k_{11}k_{20}k_{22}x_4x_6 + 
			k_9 k_{11} (k_{21}+k_{22}) (k_{23}+k_{31})x_4  \\ 
&+ &			 k_{20}k_{22}(k_{10}+k_{11})(k_{13}+k_{30})x_6 
			 +(k_{10}+k_{11})(k_{21}+k_{22})(k_{13}k_{23}+k_{23}k_{30}+k_{13}k_{31}).
\end{matrix}
$$


\[ \begin{array}{|ccr|ccr| } \hline
x_2 &=& \dfrac{k_1}{k_2} x_1 &       x_{12} &=& \frac{r(x_4,x_6)}{k_{12}k_{30}(k_{10}+k_{11})(k_{21}+k_{22})} \dfrac{k_{25}}{k_{24} }  x_{13}\\[4mm] 
x_3 &=& \dfrac{k_1 k_{26}}{k_2 k_{27}} x_1 &        x_{14} &=& \dfrac{k_1k_3}{k_2(k_4 + k_5)} x_1 x_4\\[4mm] 
x_5 &=& \dfrac{k_1k_3k_5(k_7+k_8)}{k_2k_6k_8(k_4+k_5)} \dfrac{x_1x_4}{x_8} &        x_{15} &=& \dfrac{k_1k_{14}k_{26}}{k_2k_{27} (k_{15} + k_{16})} x_1 x_6 \\[4mm] 
x_7 &=& \dfrac{k_1k_3k_5k_{28}(k_7+k_8)}{k_2k_6k_8k_{29}(k_4+k_5)} \dfrac{x_1x_4}{x_8}&       x_{16} &=& \dfrac{k_1k_3k_5}{k_2k_8(k_4 + k_5)} x_1 x_4\\[4mm] 
 x_9 &=&  \frac{k_6k_8k_{14}k_{16}k_{26}k_{29}(k_4+k_5)(k_{18}+k_{19})}
			{k_3k_4k_5k_{17}k_{19}k_{27}k_{28}(k_7+k_8)(k_{15}+k_{16})} 
			\dfrac{x_6x_8}{x_4} &        x_{17} &=& \dfrac{k_1k_{14}k_{16}k_{26}}{k_2k_{19}k_{27}(k_{15} + k_{16})} x_1x_6 \\[4mm] 
x_{10} &=& \frac{k_{12}(k_{10}+k_{11})(k_{20}k_{22}x_6 + (k_{21} + k_{22})(k_{23} + k_{31}))}{ r(x_4,x_6)}
			&        x_{18} &=& \frac{k_9k_{12}(k_{20}k_{22}x_6+(k_{21} + k_{22})(k_{23}+k_{31}))}{r(x_4,x_6)} x_4 \\[4mm] 
x_{11} &=&    \dfrac{k_{12}k_{30}(k_{10}+k_{11})(k_{21}+k_{22})}{r(x_4,x_6)} & x_{19} &=& \dfrac{k_{12}k_{20}k_{30}(k_{10}+k_{11})}{r(x_4,x_6)} x_6 \\[4mm]  \hline
\end{array} \]

\smallskip

This map $\varphi_Y$ is substituted into  (\ref{eq:conservation}), and then we
saturate. The resulting ideal  $J_Y$ equals
$$
\begin{matrix} \langle \alpha_1 x_6x_8 + \alpha_2 x_4 + \alpha_3 x_6,  & \alpha_4 x_1x_6 + \alpha_5 x_1 + \alpha_6 x_8 + \alpha_7,\\
\alpha_8 x_1x_4 + \alpha_9 x_8 + \alpha_{10},&
\alpha_{11} x_4x_6x_{13} + \alpha_{12} x_4x_{13} + \alpha_{13} x_6x_{13} + \alpha_{14} x_{13} + \alpha_{15}, \\
\alpha_{16} x_4x_6^2 + \alpha_{17} x_6^3 + \alpha_{18} x_4x_6 + & \hspace{-2mm} \alpha_{19} x_6^2 + \alpha_{20} x_8^2 
+ \alpha_{21} x_1 + \alpha_{22} x_4 + \alpha_{23} x_6 + \alpha_{24} x_8 + \alpha_{25} \rangle,
\end{matrix}
$$
where the $\alpha_1,\ldots,\alpha_{25}$
 are certain explicit rational functions in the ${\bf k}$-parameters.
\end{ex}

\section{Polyhedral Geometry}
\label{sec6}

Dynamics of the system while not at steady state cannot typically be studied with algebraic methods. One exception is the set of all possible states accessible from a given set of initial values via the chemical reactions in the model.  This set is called a {\em stoichiometric compatibility class} in the biochemistry literature.
Mathematically, these classes are convex polyhedra.
We determine them all for the Wnt shuttle model.
This resolves Problem 6 from the Introduction.

The conservation relations (\ref{eq:conservation})
 define a linear map $\chi$ from the orthant
of concentrations  $\RR_{\geq 0}^{19}$
  to the orthant of conserved quantities  $\RR_{\geq 0}^5$.
We express this projection as a $5 {\times} 19$-matrix:
\setcounter{MaxMatrixCols}{20}
\begin{equation}
\label{eq:five_nineteen}
\begin{pmatrix} c_1 \\ c_2 \\ c_3 \\ c_4 \\ c_5 \end{pmatrix} \,\, = \,\,
\begin{pmatrix}
 1 & 1 & 1 & \cdot & \cdot & \cdot & \cdot & \cdot & \cdot & \cdot & \cdot & \cdot & \cdot & 1 & 1 & \cdot & \cdot & \cdot & \cdot \\
  \cdot & \cdot & \cdot & 1 & 1 & 1 & 1 & \cdot & \cdot & \cdot & \cdot & \cdot & \cdot & 1 & 1 & 1 & 1 & 1 & 1 \\
 \cdot & \cdot & \cdot & \cdot & \cdot & \cdot & \cdot & 1 & \cdot & \cdot & \cdot & \cdot & \cdot & \cdot & \cdot & 1 & \cdot & \cdot & \cdot \\
 \cdot & \cdot & \cdot & \cdot & \cdot & \cdot & \cdot & \cdot & 1 & \cdot & \cdot & \cdot & \cdot & \cdot & \cdot & \cdot & 1 & \cdot & \cdot \\
 \cdot & \cdot & \cdot & \cdot & \cdot & \cdot & \cdot & \cdot & \cdot & \cdot & \cdot & 1 & 1 & \cdot & \cdot & \cdot & \cdot & \cdot & \cdot 
 \end{pmatrix}
 \cdot \begin{pmatrix} x_1 \\ x_2 \\ x_3 \\ \vdots \\ x_{18} \\ x_{19} \end{pmatrix}
\end{equation}

Let $P_{\bf c}$ denote the fiber of the map $\chi$ for ${\bf c} \in \RR_{\geq 0}^5$.
This is known in the biochemical literature as the
{\em invariant polyhedron} or the {\em stoichiometric compatibility class}
of the given~${\bf x}$; see e.g.~\cite[(3)]{ShSt}.
The fiber over the origin  is $\,P_{\bf 0} = \RR_{\geq 0} \{{\bf e}_{10} , {\bf e_{11}}\}$, the two-dimensional orthant formed by all positive linear combinations of ${\bf e}_{10}$ and ${\bf e}_{11}$. 
 If ${\bf c} \in \RR_{\geq 0}^5 $ is an interior point, then $P_{\bf c}$ is a $14$-dimensional 
 convex polyhedron of the form $P_{\bf 0} \times {\tilde P}_{\bf c}$
 where ${\tilde P}_{\bf c}$ is a $12$-dimensional (compact) polytope.
 Two vectors ${\bf c}$ and ${\bf c}'$ are considered {\em equivalent}
 if their invariant polyhedra $P_{\bf c}$ and $P_{{\bf c}'}$ have the same normal fan.
 This property is much stronger than being
   combinatorially isomorphic.
 The equivalence classes are relatively open polyhedral cones,
 and they define a partition of $\RR_{\geq 0}^5$. This partition
 is the {\em chamber complex} of the matrix (\ref{eq:five_nineteen}). 
 For a low-dimensional illustration, see
 \cite[Figure 1]{ShSt}. Informally speaking, the chamber complex 
 classifies the possible boundary behaviors of our dynamical system.
 
 \begin{prop} \label{prop:sixone}
 The chamber complex of our $5 {\times} 19$-matrix  divides $\RR_{\geq 0}^5$ into
 $19$ maximal cones. It is the product of a ray, $\RR_{\geq 0}$,
 and the cone over a subdivision of the tetrahedron.
 That subdivision consists of $18$ smaller tetrahedra
 and $1$ bipyramid, described in detail below.
  \end{prop}
  
  \begin{proof}
  The product structure arises because the matrix 
  has two blocks after permuting columns, an upper left $4 {\times} 17$ block and
  a lower right $1 {\times} 2$ block $( 1 \,\, 1 )$.
 Our task is to compute the chamber decomposition of
 $\RR_{\geq 0}^4$ defined by the $4 \times 17$-block.
 After deleting zero columns and multiple columns, 
 we are left with a $ 4 \times 7$-matrix, given by
  the seven left columns in
  $$ M \quad = \quad 
 \bordermatrix{  & a & b & c  & d & e & f & g \,\,\,& h & i & j & k & l \cr
  & 0 & 1 & 0 & 0 & 1 & 0 & 0 \,\,\, & 0 & 1 & 1 & 1 & 1 \cr 
   & 1 & 1 & 1 & 1 & 0 & 0 & 0 \,\,\, & 1 & 1 & 1 & 1 & 2 \cr
 & 0 &  0 & 1 & 0 & 0 & 1 & 0 \,\,\, & 1 & 0 & 1 & 1 & 1 \cr 
 & 0 & 0 & 0 & 1 & 0 & 0 & 1  \,\,\, & 1 & 1 & 0 & 1 & 1 \cr
  }.  $$
  The correspondence between the 
  seven left columns of $M$ and the columns of (\ref{eq:five_nineteen}) is as follows:
$$     \begin{matrix}
a =  \{x_4, x_5, x_6, x_{7}, x_{18}, x_{19}\}, \quad
b =   \{x_{14},x_{15}\}, \quad
c =  \{x_{16}\}, \\
d =  \{x_{17}\},   \quad
e =  \{x_1, x_2, x_3\} , \quad
f =  \{x_8\}, \quad
g =  \{x_{9}\}.
\end{matrix}
$$
The remaining columns of $M$ are additional vertices in the subdivision.

The following table lists the $19$ maximal chambers.
For each chamber we list the extreme rays and the
facet-defining inequalities. For instance, the 
chamber in $\RR_{\geq 0}^5$ denoted by
 $efjk$ is the orthant spanned by
the columns $e$, $f$, $j$ and $k$ of the matrix $M$
times the ray $(0,0,0,0,1)^T$. It is defined by
$c_5 \geq 0$ together with the four listed inequalities:
$c_4 \geq 0, \,{\rm min}(c_1,c_3) \geq c_2 \geq c_4$.
$$ \begin{matrix}
abcd &   \{c_4,c_3,c_1, c_2-c_4-c_3-c_1\} \\
bcdl  &   \{c_2-c_3-c_1, c_2-c_4-c_1, c_2-c_4-c_3, -c_2+c_4+c_3+c_1\}  \\
efgk & \{c_2, -c_2+c_4, -c_2+c_3, -c_2+c_1\} \\
bcjl &\{c_4, -c_2+c_1+c_3, c_2-c_3-c_4, c_2-c_1-c_4\} \\
bdil &  \{c_3, -c_2+c_4+c_1, c_2-c_4-c_3, c_2-c_3-c_1\} \\
beij &  \{c_3, c_4, c_1 -c_2, c_2-c_3-c_4\} \\
cdhl &  \{c_1, -c_2+c_3+c_4, c_2-c_4-c_3, c_2-c_4-c_1\} \\
cfhj & \{c_4, c_1, -c_2+c_3, c_2-c_4-c_1\}  \\
dghi &  \{c_1, c_3, -c_2+c_4, c_2-c_1-c_3\} \\
egik &  \{c_3, -c_2+c_4, c_2-c_3, -c_2+c_1\} \\
fghk &  \{c_1, -c_1+c_2, -c_2+c_3, -c_2+c_4\} \\
efjk & \{c_4, c_1- c_2, c_2-c_4, -c_2+c_3\} \\
bijl & \{c_2-c_1, -c_2+c_1+c_3, -c_2+c_4+c_1, c_2-c_3-c_4\} \\
chjl & \{c_2-c_3, -c_2+c_4+c_3, -c_2+c_3+c_1, c_2-c_4-c_1\} \\
dhil &  \{c_2-c_4, -c_2+c_4+c_1, -c_2+c_3+c_4, c_2-c_1-c_3\} \\
ghik &  \{c_4-c_2, c_2-c_3, c_2-c_1, -c_2+c_1+c_3\} \\
eijk & \{c_2-c_4, c_2-c_3, c_1-c_2, -c_2+c_3+c_4\} \\
fhjk & \{c_2-c_4, c_2-c_1, -c_2+c_3, -c_2+c_4+c_1\} \\
hijkl & \{c_2-c_4, c_2-c_3, c_2-c_1, -c_2+c_4+c_3, -c_2+c_4+c_1, -c_2+c_3+c_1\}
\end{matrix}
$$
Interpreting the columns of $M$ as homogeneous
coordinates, the table describes a subdivision of the
standard tetrahedron into $18$ tetrahedra
and one bipyramid $hijkl$.
These cells use the $12$ vertices $a,b,\ldots,l$.
The reader is invited to check that 
this subdivision has precisely $39$ edges and $47$ triangles, 
so the Euler characteristic is correct: $  12 - 39 + 47 - 19 = 1$.
\end{proof}
We shall prove the following result about the Wnt shuttle model.
\begin{prop} \label{prop:boundary}
Suppose that the rate constants $k_i$ and the conserved quantities
$c_j$ are all strictly positive. Then no steady states exist on the
boundary of the invariant polyhedron $P_{\bf c}$.
\end{prop}

\begin{proof}
Consider the two components $I_m$ and $I_e$ of the steady state ideal $I$
given in Lemma  \ref{lem:3_1}. We intersect each of the two varieties
with  the affine-linear space defined by the conservation relations (\ref{eq:conservation})
for some ${\bf c} \in \mathbb{R}_{>0}^5$.
We claim that all solutions ${\bf x}$ satisfy
 $x_i \not= 0$ for $i=1,2,\ldots,19$. 

For the main component $V(I_m)$, we prove this assertion with the help of the parametrization $\varphi_Y$ from
Example \ref{ex:param}. If the values of $x_1,x_4,x_6,x_8,x_{13}$ and of the expression $r(x_4,x_6)$ are nonzero, then each coordinate of $\varphi_Y$ is nonzero. 
We next observe that $r(x_4,x_6) > 0$ for any ${\bf k} > 0$ and ${\bf x} \geq 0$. 
A case analysis, using binomial relations in the ideal $I_m$, reveals that if any of $x_1,x_4,x_6,x_8,x_{13}$ are zero, some coordinate of ${\bf c}$ is forced to zero as well:
\[
\begin{array}{rcrcr}
x_1 = 0 & \Rightarrow & x_2,x_3,x_{14},x_{15} = 0 &  \Rightarrow & c_1 = 0, \\
x_{13} = 0 & \Rightarrow & x_{12} = 0 &  \Rightarrow & c_5 = 0, \\
x_4 = 0 & \Rightarrow & x_5,x_6, x_7, x_{14},x_{15},x_{16},x_{17},x_{18},x_{19} = 0 &  \Rightarrow & c_2 = 0, \\
 & \text{or} & x_8,x_{16} = 0 &  \Rightarrow & c_3 = 0, \\
x_6 = 0 & \Rightarrow & x_9, x_{17} = 0 &  \Rightarrow & c_4 = 0, \\
 & \text{or} & x_4 = 0 &  \Rightarrow & c_2 \text{ or } c_3 = 0, \\
x_8 = 0 & \Rightarrow & x_{16} = 0 &  \Rightarrow & c_3 = 0. \\
\end{array}
\]

It remains to consider the extinction component.
Its ideal $I_e$ contains the set 
$\,b \cup l = \{x_1, x_2, x_3, x_{14}, x_{15}\}$. The corresponding columns of the 
matrix in  (\ref{eq:five_nineteen}) are the only columns
with a nonzero entry in the fourth row. This implies that
$c_4 = 0$ holds for every steady state in $V(I_e)$.
We conclude that there are no steady states on the boundary of
the polyhedron~$P_{\bf c}$.~\end{proof}

\begin{remark}
In this proof we did not  need
the detailed description of the chamber complex,
because of the special combinatorial structure in the Wnt shuttle model.
In general, when studying chemical reaction networks that arise in systems biology,
an analysis like Proposition \ref{prop:sixone} is requisite for gaining
information about possible zero coordinates in
the steady states.
\end{remark}

\section{Parameter Estimation}
\label{sec7}

Question 7 asks: {\em What information does species concentration data give us for parameter estimation?}  This question is of particular importance to experimentalists, as species concentrations depend on initial conditions, whereas parameter values are intrinsic to the biological process being modeled.
Identifiability of parameters has been studied in many contexts, notably in
statistics \cite{GSS} and in biological modeling \cite{MS}.
  Sometimes, as in \cite{MS}, parameters are determined from complete time-course data of the dynamical system, making a differential algebra approach desirable. 
  In the present paper we focus on the  steady state variety, so we consider data collection only at steady state.
  We assume that there is a true
  but unknown parameter vector ${\bf k}^* \in \mathbb{R}^{31}$
  of rate constants, and our data  are sampled from the
  positive real points ${\bf x}$ on the variety in $\mathbb{R}^{19}$ that is defined by the
  $19$  polynomials in (\ref{eq:diffeqn}).
  
\subsection{Complete Species Information.} The first algebraic question we answer: 
To what extent is the true parameter vector ${\bf k}^*$ determined by
points on its steady state variety?

To address this question, we form the  polynomial matrix
$F({\bf x})$ of format $19 \times 31$ whose entries are the 
coefficients of the right-hand sides of (\ref{eq:diffeqn}),
regarded as linear forms in ${\bf k}$.
With this notation, our dynamical system (\ref{eq:diffeqn})
can be written in matrix-vector product form as 
$$ \dot{\bf x} \,=\, F({\bf x}) \cdot {\bf k}.$$
Our data points are sampled from 
\begin{equation}
\label{eq:PSSV}
 \bigl\{\, {\bf x} \in {\mathbb R}^{19}_{> 0} \,\,:\,\,
F({\bf x}) \cdot {\bf k}^* = {\bf 0} \,\bigr\}.
\end{equation}
Let ${\bf x}_1, {\bf x}_2 , {\bf x}_3,\ldots $ denote generic data points in (\ref{eq:PSSV}).
The set of all parameter vectors ${\bf k}$ that are compatible with these data
is a linear subspace of ${\mathbb R}^{31}$, namely it is the intersection 
\begin{equation}
\label{eq:furthermeasure}
{\rm kernel}(F({\bf x}_1)) \,\cap \, {\rm kernel}(F({\bf x}_2)) \,\cap \, {\rm kernel}(F({\bf x}_3)) \,\cap \,\cdots
\end{equation}
The best we can hope to recover from sampling data
is the following  subspace containing ${\bf k}^*$:
\begin{equation}
\label{eq:besthope}
 \bigcap_{{\bf x} \,{\rm in} \, (\ref{eq:PSSV})} \! {\rm kernel}(F({\bf x})) \,\,\, \subset \,\,\, \mathbb{R}^{31}. 
 \end{equation}
We refer to (\ref{eq:besthope}) as the  {\em space of parameters compatible with ${\bf k}^*$}.
A  direct computation reveals:

\begin{prop} 
\label{prop:nonident}
The space of all parameters compatible with ${\bf k}^*$ is a $14$-dimensional subspace
of $\mathbb{R}^{31}$. If ${\bf x}$ is generic then 
the kernel of $F({\bf x})$ is a $17$-dimensional subspace of 
$\mathbb{R}^{31}$. 
\end{prop}

This has the following noteworthy consequence for our  biological application:

\begin{cor}
The parameters of the Wnt shuttle model are not identifiable from steady state data,
but there are $14$ degrees of freedom in recovering
the true parameter vector ${\bf k}^*$.
\end{cor}

Our next step is to gain a more precise understanding of
the subspaces in Proposition \ref{prop:nonident}.
To do this, we shall return to the combinatorial setting of matroid theory.
We introduce two matroids on the $31$ reactions in Table \ref{tab-reactions}.
The common ground set  is $K = \{k_1,k_2,\ldots,k_{31}\}$.
The  {\em one-point matroid}  $\mathcal{M}_{\rm one}$ is the rank $17$ matroid on
$K$ defined by the linear subspace ${\rm kernel}(F({\bf x}))$  of $\mathbb{R}^{31}$
where ${\bf x} \in \mathbb{R}^{19}$ is generic.
The {\em parameter matroid}  $\mathcal{M}_{\rm par}$ is the rank $14$ matroid
on $K$ defined by the space (\ref{eq:besthope})
of all parameters compatible with a generic~${\bf k}^*$.
The following result,  obtained by calculations,
reflects the block structure of the matrix~$F({\bf x})$.

\begin{figure}[h]
\includegraphics[scale=.8]{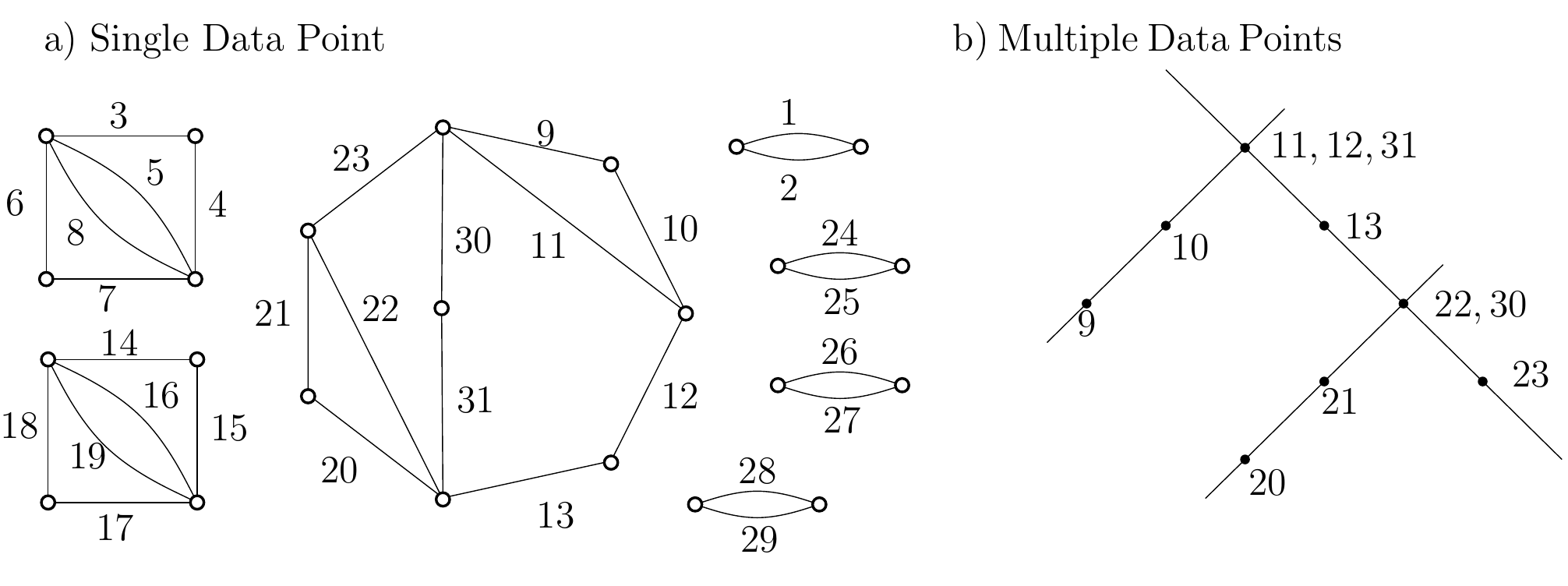}
\caption{Graphic representation of the one-point matroid  $\mathcal{M}_{\rm one}$ of rank $17$.
The rank $4$  component  of
the rank $14$ parameter matroid $\mathcal{M}_{\rm par}$ is not graphic.}
\label{fig:kmatroid}
\end{figure}

\begin{prop} 
\label{prop:kmatroid}
The one-point matroid  $\mathcal{M}_{\rm one}$ is the graphic matroid 
of the graph shown in Figure \ref{fig:kmatroid} a).
Its seven connected components are matroids of ranks $3,3,7,1,1,1,1$.
The rank $14$ parameter matroid $\mathcal{M}_{\rm par}$ is obtained
from $\mathcal{M}_{\rm one}$ by specializing the rank $7$ component
to the rank $4$ matroid on $11$ elements whose affine representation is shown in
Figure \ref{fig:kmatroid} b).
\end{prop}

This characterizes the combinatorial constraints imposed on the parameters 
${\bf k}$ by measuring the species concentrations at steady state.
For a single measurement ${\bf x}$, the result on $\mathcal{M}_{\rm one}$
tells us that the $19 \times 31$-matrix $F({\bf x})$ has rank $14 = 31 - 
{\rm rank}(\mathcal{M}_{\rm one})$.
After row operations, it block-decomposes into two matrices
of format $3\times 6$, one matrix of format $4 \times 11$, 
and four matrices of format $1 \times 2$.
Each of these seven matrices is row-equivalent to the
node-edge {\em cycle matrix} of a directed graph,
with underlying undirected graph as
in Figure \ref{fig:kmatroid} (a).

Consider the graph with edges
$9,10,11,12,13,20,21,22,23,30,31$.
The cycle $\{22,23,30,31\}$ reveals that our measurement ${\bf x}$
imposes \underline{one} linear constraint on
$k_{22}, k_{23},k_{30},k_{31}$.
If we take further measurements, as in
(\ref{eq:furthermeasure}), then six of the seven blocks
of $F({\bf x})$ remain unchanged. Only the
$4 \times 11$-block of $F({\bf x})$ 
must be enlarged, to a $7 \times 11$-matrix.
The rows of that new matrix specify the affine-linear dependencies
among $11$ points in $\mathbb{R}^3$.
That point configuration is depicted in 
Figure \ref{fig:kmatroid} (b). For instance,
the points $\{9,10,11\}$ are collinear,
the points $\{20,21,22\}$ are collinear,
but these two lines are skew in $\mathbb{R}^3$.
From the other line we see that that
repeated measurements at steady state 
impose \underline{two} linear constraints on
$k_{22}, k_{23},k_{30},k_{31}$.

\subsection{Circuit Data.} The second question we address in this section: {\em Given partial species concentration data, is any information about parameters available?} In Section 7.1,
all $19$  concentrations $x_i$ were available for a steady state.
In what follows, we suppose that $x_i$ can only be measured for indices $i$ in a subset of
the species, say $C \subset \{1,\ldots,19\}$. In our analysis, it will be useful to take
advantage of the rank $5$ algebraic matroid in Proposition \ref{prop:951},
since that matroid governs  dependencies among
the coordinates $x_1,\ldots,x_{19}$ at steady states.

We here focus on the special case when $C$ is one
of the $951$ circuits of the algebraic matroid of $\widetilde{I_m}$.
Let $f_C$ be the corresponding circuit polynomial, as 
in Table \ref{tab-circuits}.
We regard $f_C$ as a polynomial in ${\bf x}$ whose coefficients
are polynomials in $\mathbb{Q}[{\bf k}]$. Suppose that $f_C$
has $r$ monomials ${\bf x}^{a_1} ,\ldots , {\bf x}^{a_r}$. 
We write
$F_C  \in \mathbb{Q}[{\bf k}]^r$ for the vector of coefficients, so
our circuit polynomial is the dot product
$\,f_C({\bf k},{\bf x}) = F_C ({\bf k}) \cdot ({\bf x}^{a_1} ,\ldots , {\bf x}^{a_r})$.
We write $\mathcal{V}_C \subset \mathbb{R}^r$ for the algebraic
variety parametrized by $F_C({\bf k})$.
Thus $\mathcal{V}_C$ is the Zariski closure in
$\mathbb{R}^r $  of the set
$\{F_C({\bf k}')\,:\,{\bf k}' \in \mathbb{R}^{31}\}$.

Our idea for parameter recovery is this:
rather than looking for ${\bf k}$ compatible with the
true parameter ${\bf k}^*$, we seek a point ${\bf y} = F_C({\bf k})$ in $\mathcal{V}_C$
that is compatible with $F_C({\bf k}^*)$. And, only later do we 
compute a preimage of ${\bf y}$ under the map
$ \mathbb{R}^{31} \rightarrow \mathbb{R}^r $ given by $F_C$.
Most interesting is the case when  $\mathcal{V}_C$ is a
proper subvariety of $\mathbb{R}^r$.
Direct computations yield the following:

\begin{prop}
\label{prop:288}
For precisely $288$ of the $951$ circuits $C$ of the algebraic matroid of the steady state ideal
$\widetilde{I_m}$, the coefficient variety $\mathcal{V}_C$ is a proper subvariety
in its ambient space $\mathbb{R}^r$.
In each of these cases, the defining ideal of $\mathcal{V}_C$ 
is of one of the following four types:
\begin{align} \langle y_2y_6-y_3y_5 \rangle \label{eq:implicitIdeal1} \\
\langle y_5y_6-2y_3y_7,y_5^2-4y_2y_7,y_3y_5-2y_2y_6,y_2y_6^2-y_3^2y_7 \label{eq:implicitIdeal2} \rangle  & \\
\langle y_3y_5^2-y_2y_5y_6+y_1y_6^2 \rangle \label{eq:implicitIdeal3}  & \\
\langle 2y_3y_4-y_2y_5,y_2y_3-2y_1y_5,y_2^2-4y_1y_4 \rangle \label{eq:implicitIdeal4} &
\end{align}
\end{prop}

\begin{ex} \label{ex:6_10_18}
Consider the circuit $C =  \{ 6,10,18\}$. The circuit polynomial $f_C$ equals
\[ \begin{matrix}
(k_{13}k_{20}k_{22}+k_{20}k_{22}k_{30}) \cdot x_6x_{10}
\,+\,k_{11}k_{20}k_{22} \cdot x_6x_{18}
\,-\,k_{12}k_{20}k_{22} \cdot x_6 \\ +(k_{13}k_{21}k_{23}+k_{13}k_{22}k_{23}+k_{21}k_{23}k_{30}+k_{22}k_{23}k_{30}+k_{13}k_{21}k_{31}+k_{13}k_{22}k_{31}) \cdot x_{10} \\ +(k_{11}k_{21}k_{23}+k_{11}k_{22}k_{23}+k_{11}k_{21}k_{31}+k_{11}k_{22}k_{31}) \cdot x_{18} \\ -\,(k_{12}k_{21}k_{23}+k_{12}k_{22}k_{23}+k_{12}k_{21}k_{31}
+k_{12}k_{22}k_{31}). \phantom{xxx}
\end{matrix} \]
Here $r = 6$ and we write $F_C({\bf k}) = (y_1,y_2,y_3,y_4,y_5,y_6)$ for the
vector of coefficient polynomials. The variety $\mathcal{V}_C$ is the
hypersurface in $\mathbb{R}^5$ defined by the equation
$y_2 y_6 = y_3 y_5$.
\end{ex}

We now sample data points ${\bf x}_i$ from the
model with the true (but unknown) parameter vector ${\bf k}^*$.   Each such point
defines a hyperplane 
$\,\{ {\bf y} \in \mathbb{R}^r \,:\,{\bf y} \cdot ({\bf x}_1^{a_1} ,\ldots , {\bf x}_r^{a_r}) = 0 \}$.
The parameter estimation problem is to find the 
intersection of  these data hyperplanes with the variety $\mathcal{V}_C$.
That intersection contains the point ${\bf y}^* = F_C({\bf k}^*)$,
which is what we now aim to recover.

\subsection{Noisy Circuit Data.}
The final question we consider in this section is: {\em Given partial species concentration data with noise, is any information about parameters available?} 

As in Section 7.2, we fix
a circuit $C$ of the algebraic matroid in Section \ref{sec5},
and we assume that we can only measure the concentrations $x_j$
where $j \in C$. Each measurement ${\bf x}_i \in \mathbb{R}^C$ still defines a hyperplane
$\,{\bf y} \cdot ({\bf x}_i^{a_1} ,\ldots , {\bf x}_i^{a_r}) = 0 \,$
in the space $\mathbb{R}^r$. But now the true vector
${\bf y}^* = F_C({\bf k}^*)$ is not exactly on that hyperplane,
but only close to it. Hence, if we take $s$ repeated measurements,
with $s > r$, the intersection of these hyperplanes should be empty.

We propose to find the best fit by solving the following least squares optimization problem:
\begin{equation}
\label{eq:leastsquares}
 {\rm Minimize} \quad \sum_{i=1}^s \bigl(\,{\bf y} \cdot ({\bf x}_i^{a_1} ,\ldots , {\bf x}_i^{a_r})\,\bigr)^2 
\quad \hbox{subject to} \,\,\, {\bf y} \in \mathcal{V}_C \,\cap\, \mathbb{S}^{r-1} ,
\end{equation}
where $\mathbb{S}^{r-1} = \{ {\bf y} \in \mathbb{R}^r : y_1^2 + y_2^2 + \cdots + y_r^2 = 1\}$ denotes the unit sphere.
When the variety $\mathcal{V}_C$  is the full ambient space $\mathbb{R}^r$,
this is a familiar regression problem, namely, to find the hyperplane through the
origin that best approximates  $s$ given points in $\mathbb{R}^r$.
Here ``best'' means that the sum of the squared distances of the $s$ points
to the hyperplane is minimized.
This happens for $663$ of the $951$ circuits $C$, and in that case
we can apply standard techniques.

However, for the $288$ circuits $C$ identified in Proposition \ref{prop:288},
the problem is more interesting. Here the hyperplanes
under consideration are constrained to live in
a proper subvariety.
 In that case we need some algebraic geometry
to reliably find the global optimum in (\ref{eq:leastsquares}).

Our problem is to minimize a quadratic function over the real affine variety $\mathcal{V}_C \cap
\mathbb{S}^{r-1}$. The quadratic objective function is generic because
the ${\bf x}_i$ are sampled with noise.
The intrinsic algebraic complexity of our optimization problem was studied 
by Draisma et al.~in~\cite{DHOST}.
That complexity measure is the {\em  ED degree} of $\mathcal{V}_C \cap \mathbb{S}^{r-1}$,
which is the number of solutions  in $\mathbb{C}^r$
to the critical equations of (\ref{eq:leastsquares}).
Here, by ED degree we mean the ED degree of $\mathcal{V}_C \cap \mathbb{S}^{r-1}$,
when considered in generic coordinates. This was 
called the {\em generic ED degree} in \cite{OSS}.

We illustrate our algebraic approach 
by working out the first instance
(\ref{eq:implicitIdeal1}) 
in Proposition~\ref{prop:288}.

\begin{ex}
Suppose we are given $s$ noisy measurements of the concentrations $x_6, x_{10},x_{18}$.
In order to find the best fit for the parameters ${\bf k}$, we 
employ the circuit polynomial $f_C$ in Example \ref{ex:6_10_18}.
We compute ${\bf y} \in \mathbb{R}^6$ by solving the corresponding optimization problem
(\ref{eq:leastsquares2}). This problem is to minimize a random quadratic form
subject to two quadratic constraints
\begin{equation}
\label{eq:leastsquares2}
 y_2y_6-y_3y_5  \,\,= \,\,
 y_1^2 + y_2^2 + y_3^2 + y_4^2 + y_5^2 + y_6^2  -1 \,\,=\,\,0.
\end{equation}
We solve this problem using the method of Lagrange multipliers.
This leads to a system of polynomial equations in ${\bf y}$. Using saturation, we
remove the singular locus of (\ref{eq:leastsquares2}),
which is the circle $\{{\bf y} \in \mathbb{R}^6: y_1^2+y_4^2-1 = y_2 = y_3 = y_5 = y_6 = 0\}$. The resulting
ideal has precisely $40$ zeros in $\mathbb{C}^{6}$. 
In the language of \cite{DHOST,OSS}, the generic ED degree of the variety
(\ref{eq:leastsquares2}) equals~$40$.
\end{ex}

\section{From Algebra to Biology}
\label{sec8}

The aims of this paper are: (1) to demonstrate how biology can lead to interesting
questions in algebraic geometry, and (2) to apply new techniques from computational algebra
in biology. So far, our tour through (numerical) algebraic geometry, polyhedral geometry and combinatorics
has demonstrated the range of mathematical questions to explore. In this section, we will focus on
translating our analysis into applicable considerations for the research cycle in systems biology,
which is illustrated in Figure~\ref{fig:sysbio}.
In what follows we discuss some concrete applications and results pertaining to 
the steps (a), (b) and (c) in 
Figure~\ref{fig:sysbio}.


\begin{figure}[ht!]
\begin{center}
	\includegraphics[width=6.7in]{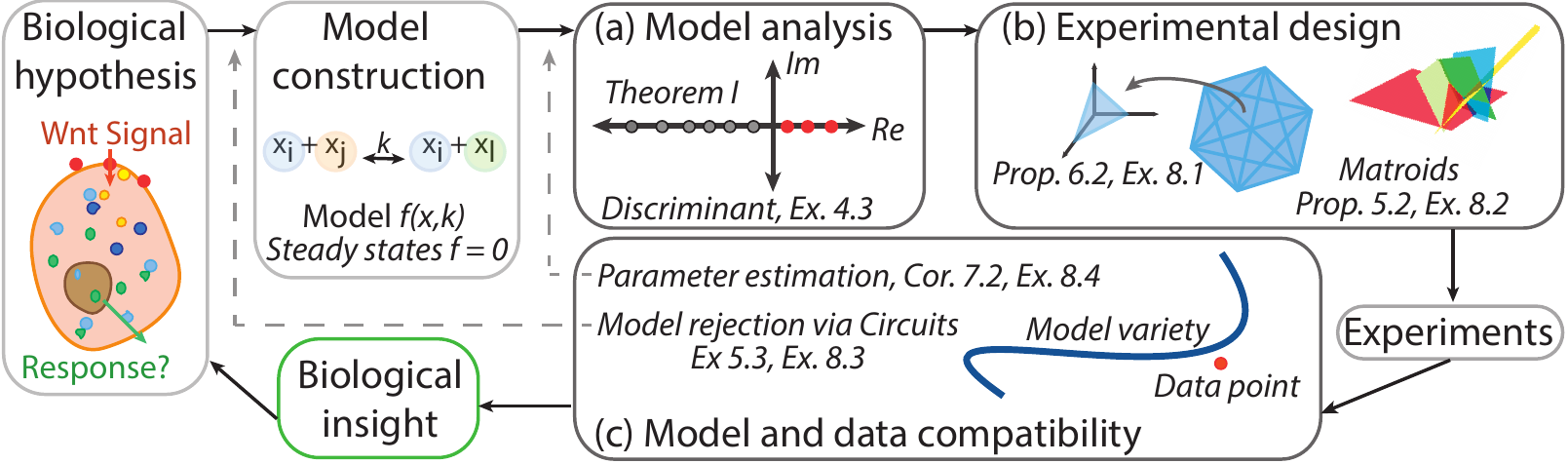}
	\caption{Systems biology cycle informed by algebraic geometry and combinatorics. (a) Model analysis. See Sections 1, 3, 4. (b) Experimental design. See Sections 5 and 6. (c) Model and data compatibility. See Sections 5 and 7. }
	\label{fig:sysbio}
\end{center}
\end{figure}

{\bf Analysis of the Model}:  Before any experiments are performed, our techniques inform the
modeler of the global steady-state properties of the model.
The number of real solutions to system (\ref{eq:diffeqn})--(\ref{eq:conservation}), stated in Theorem~\ref{thm:nine},
governs the number of observable steady states.
Various sampling schemes demonstrated that \emph{most} parameter values
lead to only one observable steady state.
We produced a set of parameter values and conserved quantities with three real solutions,
and two solutions are also attainable. 
If the ``true" parameters ${\bf k}^*$ and ${\bf c}^*$ admit multiple real 
solutions, then multistationarity of the system is theoretically possible.

If multiple states are observed experimentally, then the model must be capable of multistationarity. In the Wnt shuttle model, the system is capable of multiple steady states; however, based on parameter sampling, the frequency of this occurrence is low, and parameters in this regime are 
somewhat stable under perturbation.
The discriminant of the system is a polynomial of degree $34$ in ${\bf c}$,
and our analysis along a single line in ${\bf c}$-space illustrates the high degree of 
complexity inherent in the full stratification of the $36$-dimensional parameter space.

{\bf Experimental Design}: 
In Section \ref{sec6}, the combinatorial structure of the various stoichiometric compatibility classes 
was fully characterized. As the conserved quantities ${\bf c} = (c_1,\ldots,c_5)$ range over all
positive real values, the set of all compatible species-concentration vectors
${\bf x}$ will take one of $19$ polyhedral shapes $P_{\bf c}$.
This may find application in identifying multiple steady state solutions
for specific rate constants ${\bf k}$. A natural choice for initial conditions when 
performing experiments is on or near the vertices of the $14$-dimensional polyhedron~$P_{\bf c}$.

\begin{ex}
Suppose the conserved quantities vector  lies in the bipyramid, e.g. ${\bf c} = (1,2,2,2,3)$. 
The preimage of ${\bf c}$ in ${\bf x}$-space is  a product of the orthant
 $\RR_{\geq 0} \{ {\bf e}_{10},{\bf e}_{11}\}$ and a 12-dimensional polytope with $400$ vertices: $(1,0,0,2,0,0,0,2,2,0,0,3,0,0,0,0,0,0,0,0)$, and $399$ of its permutations.
 This product is the polyhedron $P_{\bf c}$.
 If we have control over initial conditions, beginning near the vertices positions us to find 
 interesting systems behavior.
  \end{ex}

In the laboratory, the experimentalist makes choices 
of what to measure and what not to measure. 
For instance, measuring a particular $x_i$ may be infeasible, or
there may be a situation in which measuring
concentration $x_i$ can preclude 
measuring concentration $x_j$. 

For every strategy, we fix a \emph{cost vector}, listing the costs of making 
each measurement.
We use the symbol $N$ to indicate infeasible measurements. 
Suppose there are two different ways to run the experiment; then we have a 
$2 \times 19$ \emph{cost matrix} $P$, whose rows are cost vectors for each experiment. We multiply $P$ by the
$0$-$1$-incidence matrix for the $951$ circuits of Proposition \ref{prop:951}.
That matrix has a $1$ in row $i$ and column $j$ if circuit
$j$ contains species $i$, and $0$ otherwise. The product
is a matrix of size $2 \times 951$. For $N \to \infty$, the
$2 \times 951$ matrix has a finite entry in position $(i,j)$ precisely
when the strategy $i$ can measure the circuit $j$. Minimizing over those
finite cost entries selects the most cost-effective experiment to measure a circuit.

\begin{ex}
Suppose that none of the intermediate complexes  $x_{13},\ldots,x_{19}$  are measurable, and that 
we are able to measure only one Phosphatase concentration ($x_4$ or $x_8$) in each experimental setup.
A corresponding cost matrix might look like
\[ P \,\, = \,\, \left[
\begin{array}{ccccccccccccccccccc}
1 & 1 & 1 & N & 1 & 1 & 1 & 1 & 1 & 1 & 1 & 1 & N & N & N & N & N & N & N \\
1 & 1 & 1 & 1 & 1 & 1 & 1 & N & 1 & 1 & 1 & 1 & N & N & N & N & N & N & N 
\end{array} \right]
\]
Multiplying by the circuit support matrix of size $19 \times 951$ reveals $82$ feasible experiments: 
$50$ using the first row of $P$, and $32$ using the second. With more refined cost assignment,
this would decide not only feasibility but also optimal cost. 
In this way, the matroid allows us to choose cost-minimal experiments
to obtain meaningful information for the model.
\end{ex}

{\bf Model and data compatibility}: After an experiment is performed, the task of the modeler is to test the data with the model.  One possible outcome is {\em model rejection}.
If the data are compatible, then another outcome is {\em parameter estimation}. Both may provide
 insights for biology.  
The role of algebraic geometry is
 seen in \cite{gross15,HHTS} and  shown in the next two examples.

\begin{ex}[Model Rejection]
Suppose that rate parameters $k_i$ are all known to be $1$, and that we have collected
data for variables $x_1,x_4,x_{14}$. The circuit polynomial is 
$k_1k_3x_1x_4+(-k_2k_4-k_2k_5)x_{14} $, which specializes to $ x_1x_4 - 2x_{14}$.
If the evaluation of the positive quantity $|x_1 x_4 - 2 x_{14}|$
lies above a threshold $\epsilon$, then we can reject the model as not matching the data.
\end{ex}
Every circuit polynomial of the matroid is a {\em steady state invariant}; depending on which experiment was performed, the collection of measured variables must contain some circuit. Even if one can measure all $19$ species at steady state, it is not possible to recover all 31 kinetic rate constants, but we do have relationships that must be satisfied among parameters~\cite{MHSB}.

\begin{ex}[Parameter Estimation] \label{ex:paraesti}
Suppose that rate parameters are unknown, and that we have collected
data for $x_6,x_{10},x_{18}$. The corresponding circuit polynomial
$f_C$ is shown in Example \ref{ex:6_10_18}. We know that the coefficients of $f_C$
satisfy the constraint  $y_2 y_6 = y_3 y_5$.
Suppose our experiments lead to the following ten measurements for the vector $(x_6,x_{10},x_{18})$:
\begin{small}
$$ \begin{matrix} 
\{(.715335, 4.06778, 14.6806), 
(.390982, 4.83152, 6.08251), 
(.706539, 4.98107, 3.83617),  \\ \!
(.14316, 4.30851, 12.5809),  
(.995583, 4.01222, 15), 
(.413817,4.08114, 14.902), 
(.232206, 3.38274, 23.3162), \\
(.219045, 5.06008, 3.67175), 
(.704106, 3.52804, 21.1037),  
(.648732, 3.6505, 19.7008)\}
 \end{matrix}
  $$
  \end{small} 
The data lead us to the following function to optimize in (\ref{eq:leastsquares}):
$$ \! \begin{matrix}
57.2345y_1^2+376.181y_1y_2+801.672y_2^2-27.5625y_1y_3-96.4429y_2y_3 \\
+3.36521y_3^2+179.49y_1y_4+564.034y_2y_4- 42.729y_3y_4+178.839y_4^2+564.034y_1y_5\\
+2424.31y_2y_5-144.7y_3y_5+1054.49y_4y_5+2263.2y_5^2-42.729y_1y_6\\
-144.7y_2y_6+10.339y_3y_6-83.8072y_4y_6-269.749y_5y_6+10y_6^2
  \end{matrix}
  $$
\end{ex}
The global minimum of this quadratic form on 
the codimension $2$ variety (\ref{eq:leastsquares2}) has coordinates
$$ 
 y_1 = 0.183472, \,
  y_2 = 0.152416, \,
 y_3 = 0.959232, \,
 y_4 = 0.038042, \,
 y_5 = 0.00335267, \,
 y_6 = 0.211.
$$
Given these values, one now has three degrees of freedom in estimating the
nine parameters $k_i$ that appear in the circuit polynomial $f_C$. The other ten
coordinates of ${\bf k}$ are unspecified.

\medskip

\section*{Acknowledgements}

This project was supported by UK Royal Society
International Exchange Award   2014/R1 IE140219.
EG, BS and HAH initiated discussions at an American Institute of Mathematics workshop
in Palo Alto.
Part of the work was carried out at the
Simons Institute for  Theory of Computing in Berkeley.
HAH gratefully acknowledges EPSRC Fellowship EP/K041096/1.
EG, ZR and BS were also 
supported by the US National Science Foundation, through 
grants DMS-1304167, DMS-0943745 and DMS-1419018
respectively.
Thanks to Helen Byrne and Reinhard Laubenbacher
for comments on early drafts of the paper.

\bigskip

\end{document}